\documentclass[a4paper,noarxiv,onecolumn,accepted=2025-01-16]{quantumarticle}

\usepackage[utf8]{inputenc}
\usepackage{amsmath,amssymb,mathtools,amsthm,stmaryrd}
\usepackage{physics}
\usepackage{algorithm}
\usepackage[noend]{algorithmic}
\usepackage[group-separator={,}]{siunitx}
\usepackage{todonotes}
\usepackage{booktabs}
\usepackage{afterpage}

\usepackage{comment}
\usepackage{subfig}
\usepackage{cite}
\usepackage[numbers,sort&compress]{natbib}
\usepackage{tikz}
\usetikzlibrary{quantikz}

\usepackage{xspace}
\usepackage{hyperref}
\usepackage{authblk}

\usepackage{xcolor,color}

\newcommand{\Hl}{\mathcal H}
\newtheorem{lemma}{Lemma}
\newtheorem{corollary}{Corollary}

\newcommand{\polylog}{\operatorname{polylog}}

\usetikzlibrary{decorations.pathreplacing}

\title{Prog-QAOA: Framework for resource-efficient quantum optimization through classical programs}

\author[1,2]{Bence Bak\'o}
\author[3,4]{Adam Glos}
\email{adamglos92@gmail.com}
\author[3,4]{\"Ozlem Salehi}
\author[1,2,4, 5]{Zolt\'an Zimbor\'as}
\affil[1]{Wigner Research Centre for Physics, Budapest, Hungary}
\affil[2]{Eötvös University,  Budapest, Hungary}
\affil[3]{Institute of Theoretical and Applied Informatics, Polish Academy of
	Sciences}
\affil[4]{Algorithmiq Ltd, Kanavakatu 3C 00160 Helsinki, Finland}
\affil[5]{QTF Centre of Excellence, Department of Physics, University of Helsinki, Helsinki, Finland}

\date{}

\DeclareMathOperator*{\poly}{poly}

\newcommand{\RR}{\mathbb{R}}

\newcommand{\tildorder}[1]{\tilde{\mathcal O}(#1)}

\newtheorem{theorem}{Theorem}

\begin{document}

\maketitle
\begin{abstract}
Current state-of-the-art quantum optimization algorithms require representing the original problem as a binary optimization problem, which is then converted into an equivalent cost Hamiltonian suitable for the quantum device. Implementing each term of the cost Hamiltonian separately often results in high redundancy, significantly increasing the resources required. Instead, we propose to design classical programs for computing the objective function and certifying the constraints, and later compile them to quantum circuits, eliminating the reliance on the binary optimization problem representation. This results in a new variant of the Quantum Approximate Optimization Algorithm (QAOA), which we name the Program-based QAOA (Prog-QAOA). We exploit this idea for optimization tasks like the Travelling Salesman Problem and Max-$K$-Cut and obtain circuits that are near-optimal with respect to all relevant cost measures, e.g., number of qubits, gates, and circuit depth. While we demonstrate the power of Prog-QAOA only for a particular set of paradigmatic problems, our approach is conveniently applicable to generic optimization problems.
\end{abstract} 

\section{Introduction}


Noise-robust quantum hardware holds the promise of solving many problems much faster than classical computers. Among the best examples, one could mention Grover's search providing a quadratic speed-up over classical brute-force search~\cite{grover1996fast} or Shor's algorithm for solving the factorization problem in polynomial time~\cite{shor1999polynomial}. However, quantum devices that could run these demanding quantum algorithms for useful instances do not exist yet; instead, we have only access to small quantum computers with a limited coherence time. This started the Noisy Intermediate-Scale Quantum era~\cite{preskill2018quantum}, in which a huge effort is underway in order to use these devices for real or realistic problems, which include optimization \cite{ peruzzo2014variational, albash2018adiabatic, farhi2014quantum}, machine learning~\cite{biamonte2017quantum}, and quantum simulation~\cite{yuan2019theory}. In particular, optimization algorithms are expected to be among the first to provide quantum advantage. For combinatorial optimization problems, the two leading algorithms are Quantum Annealing~\cite{albash2018adiabatic} and Quantum Approximate Optimization Algorithm (QAOA)~\cite{farhi2014quantum}, both motivated by the quantum adiabatic theorem~\cite{born1928beweis}.

Both Quantum Annealing and QAOA use the dynamics obtained from two non-commuting Hamiltonians, a cost Hamiltonian, and a mixer Hamiltonian. The former Hamiltonian describes the original optimization problem to be solved~\cite{lucas2014ising}, while the latter is responsible for the amplitude transition between the solutions to the problem. In Quantum Annealing, one starts with the ground state of the mixer Hamiltonian and evolves it slowly according to a time-dependent Hamiltonian interpolating between the mixer and the cost Hamiltonians. The initial quantum state is transformed into the ground state of the cost Hamiltonian, provided that the annealing time is sufficiently large.

QAOA is an alternative to Quantum Annealing designed for the quantum gate-based model~\cite{farhi2014quantum}. The algorithm is inspired by the trotterizaton of the adiabatic evolution but goes beyond that by adjusting the evolution time of the alternatingly applied mixer and cost Hamiltonians. Since the trotterized adiabatic evolution is a special case of this scheme, one can guarantee that the ground state is reachable for sufficiently many alternating evolution steps. Recent results using 289-qubit quantum hardware showed an improvement of QAOA over simulated annealing~\cite{ebadi2022quantum}. Further experimental results on QAOA are demonstrated in~\cite{PhysRevX.13.041052,shaydulin2023evidence}. Since sampling from the QAOA's
quantum circuit has been proven to be computationally hard for classical
computers~\cite{farhi2016quantum}, the scheme is also a candidate for near-term quantum computational advantage experiments. 
Finally, Boulebnane and Montanaro recently obtained analytical bounds on the average success probability of QAOA for $K$-SAT problems and showed that it has a potential for quantum speed-up~\cite{boulebnane2022solving}.

Implementation of QAOA involves constructing a circuit consisting of alternating layers of parametrized phase and mixer operators corresponding to cost and mixer Hamiltonians, respectively. The evolution times of the operators are determined by the parameters of the circuit, which are optimized by a classical algorithm to minimize the cost. Traditionally, the mixer operator consists of single-qubit $X$ rotations, but alternative mixer operators have been explored for different problem classes in later works \cite{hadfield2019quantum, bartschi2020grover, wang2020x}. On the other hand, the phase operator is problem-specific, and the current state-of-the-art way of implementation requires obtaining the dedicated model representation of the optimization problem. This involves mathematically expressing the problem and then formulating it as a Quadratic Unconstrained Binary Optimization (QUBO)~\cite{lucas2014ising}, or a Higher-Order Binary Optimization (HOBO) problem~\cite{glos2020space,tabi2020quantum}. The binary model is then transformed into an Ising model by replacing each bit variable $b_i$ with a spin variable $s_i$ by substituting $b_i=\frac{1-s_i}{2}$. Finally, the phase operator is obtained by term-wise implementation of the cost Hamiltonian through gates generated by $ZZ$ terms. One drawback of this process is that it is not always straightforward to come up with the Ising model representation of the problem at hand. A more significant downside of this approach is that it can result in circuits with an excessive number of gates. We can readily see this by a simple example: the HOBO formulation $-\prod_{i=0}^{n-1}b_i$ results in the Ising model $-\frac{1}{2^n}\prod_{i=0}^{n-1}(1-s_i)$ that has an exponential number of terms, and thus gates, as each term is implemented independently.

Taking into account the aforementioned drawbacks of the state-of-the-art method, in this paper, we provide an alternative approach for implementing QAOA. Our proposed approach, called Program-based QAOA (Prog-QAOA), deviates from the conventional term-wise implementation of the phase operator derived from the Ising model representation of the problem. Instead, we offer a novel methodology that eliminates the need for such representation. Starting with an abstract description of the problem, we first identify the objective and the constraints that need to be satisfied. For the objective, we design a classical program that computes the objective value, while for each constraint, we design a program that outputs 0 if the constraint is satisfied and a positive number otherwise. Subsequently, the programs are compiled into quantum circuits and combined with a rotation around the $ Z $ axis to achieve an equivalent effect to term-by-term implementation. The Prog-QAOA framework allows for a direct mapping from abstract problem descriptions to quantum circuits, eliminating the reliance on the QUBO/HOBO or Ising model representation. With this new approach, it is sufficient that the programs are efficiently implementable on quantum hardware, which also increases the pool of constraints that can be efficiently implemented. 

The concept of Prog-QAOA has not been explored before as a general QAOA framework, up to our knowledge. An implicit example appears in \cite{fuchs2021efficient} for the weighted Max-$K$-Cut problem, where the authors propose a quantum circuit that directly implements the problem Hamiltonian into the circuit. Notwithstanding, we will show that a more efficient implementation is possible for Max-$K$-Cut later in the paper with simpler reasoning on how the circuit can be constructed. In \cite{de2019knapsack}, the authors solve a battery optimization problem that is equivalent to the Knapsack problem using QAOA, where they compute the penalties resulting from inequalities inside the circuit instead of using the conventional penalty method.  However, this is demonstrated only in the case of linear inequalities that exist in the specific problem considered. Recently, in~\cite{boulebnane2023peptide}, it is suggested that the phase operator for a problem from quantum chemistry can be constructed by implementing arithmetics operations within a quantum circuit.

We claim that Prog-QAOA not only facilitates designing quantum circuits compared to the state-of-the-art implementation but also enables the design of resource-efficient quantum circuits with respect to cost measures such as number of qubits, number of gates, depth both on all-to-all and linear nearest neighbor architectures, number of shots required for energy estimation, and the search space size. All of those factors have an essential impact on the quality of the solution and optimization time. We demonstrate the power of Prog-QAOA through paradigmatic problems like the Travelling Salesman Problem (TSP) and Max-$K$-Cut. Our designs not only improve the known alternatives but also reach the lower bounds simultaneously for all cost measures up to a polylogarithmic factor, resulting in near-optimal circuit designs. In particular, for TSP, we go below the so far unsurpassed bound $\order{ n^3}$ for the number of gates, obtaining at the same time $\order{ n \log n}$ number of qubits and $\tildorder{n}$ depth. In addition to a detailed analysis of the aforementioned problems, we demonstrate the Prog-QAOA approach for TSP with time windows, general Integer Linear Program, and special cases of graph isomorphism. In general, any problem in NP has a classical polynomial algorithm for certifying the correctness of its solutions which can serve as the starting point for the classical program. 

The rest of the paper is organized as follows. In Sec.~\ref{sec:prog-qaoa}, we introduce the Prog-QAOA framework. In Sec.~\ref{sec:applications}, we demonstrate Prog-QAOA for Max-$K$-Cut and TSP problems followed by some technical details of implementation in Sec.~\ref{sec:technical}. In Sec.~\ref{sec:cost}, we analyze the presented algorithms taking into account various cost metrics and we perform a numerical investigation on the performance of Prog-QAOA. Finally, in Sec.~\ref{sec:discussion} we discuss our results and present conclusions with possible future investigation directions.

\section{Prog-QAOA} \label{sec:prog-qaoa}
In every optimization problem, we start with an abstract non-mathematical description. For example, the Travelling Salesman Problem (TSP) can be stated as finding the shortest closed route that visits all cities, given the travel costs between them. Starting with this initial description, a particular approach to tackle such problems involves constructing a mathematical model consisting of an objective function $ f(x) $ and usually a set of constraints to be met. In the case of TSP, the objective function is the total cost of the route, while the constraints restrict solutions to permutations of cities only. In essence, any optimization problem can be mathematically represented as a combination of an objective function and a set of constraints. Although commonly encountered constraints are in the form of (linear) equalities or inequalities,  the form of constraints $C_i$ is not limited to those. In general, one can consider $C_i: x \mapsto \{F,T\}$ such that $C_i(x) = T$ iff $x$ is a solution that satisfies $C_i$. Note that infeasible solutions are not necessarily required to violate all the constraints simultaneously.

Whether classical or quantum, to utilize generic algorithms like the simplex method, quantum annealing, or QAOA,  it is typically necessary to provide an explicit formulation tailored to the chosen method. For example, the simplex method requires a representation that consists of a linear objective function and linear equality or inequality constraints. For the quantum optimization framework, the current requirements necessitate expressing the problem in the following form:
\begin{equation}
	h(s) = Af(s) + A_1 f_{C_1}(s) + \dots + A_m f_{C_m}(s), 
\end{equation}
where $f$ is the objective function and for a given $s$ where $s$ is a vector of spin variables, $f_{C_i}(
s) \geq 0 $ are \textit{constraint functions} such that $f_{C_i}(s)=0$ if and only if $C_i$  is satisfied. $A$ is taken to be 1 in general, while $A_i$ are sufficiently large constants. For QAOA and quantum annealing, the function has to be polynomial with a polynomial number of terms, which might also be the case for unbounded degree polynomials~\cite{glos2020space}. In the case of existing quantum annealers, this function is required to be a quadratic polynomial. In both cases, the function from the equation above contributes to the quantum state construction by adding a local phase to the respective solution, i.e. $h:\ket{s} \to \exp\left (i t h(s) \right) \ket{s}$. Following the intuition from quantum annealing, which is an appropriate analogy for time evolutions with small $t$, the smaller the values of $s$ the higher the value of the amplitude corresponding to $\ket{s}$ is expected in the final quantum state. It is worth noting that there are methods that encode constraints via proper initial state preparation and/or choice of the mixer Hamiltonian, however, even for the relatively simple constraints, the cost of implementing initial states or mixer might be prohibitive~\cite{bartschi2020grover,wang2020x,hadfield2019quantum}.
 
For QAOA, we can leverage the fact that the quantum gate-based model is computationally universal and not bound to a particular layout. Therefore, instead of explicitly having $f_{C_i}$ as a polynomial, we can design a classical program that outputs 0 if a given solution satisfies constraint $C_i$ and a positive value otherwise. Such a program is then compiled into a quantum circuit, which when combined with a rotation around the $Z$ axis and uncomputation, achieves an equivalent effect to term-by-term implementation. Similarly for the objective function $ f $, we can design a classical program that computes the objective value instead of having $ f $ as a polynomial. The same program can also be used to compute the energy of the sample.

To better illustrate this approach, let us consider a constraint $C$ asserting $b_1,\dots, b_k$ form a valid one-hot vector, i.e., exactly one bit equals $1$. For simplicity, we will use bits instead of spins for convenience, but the same conclusions also apply to the latter. Such constraint is usually modeled through a penalty function
\begin{equation}
	f_C(b_1,\dots,b_k) \coloneqq \left ( \sum_{i=1}^k b_i - 1 \right )^2.
\end{equation}
The corresponding Ising model is a quadratic polynomial with $\order{k^2}$ 2-local terms and requires $\order{k^2}$ gates to be implemented on the quantum hardware. 

However, it is possible to design a very simple circuit that enforces the same constraint $C$. We can achieve this by computing the sum $\sum_{i=1}^k b_i$ on a separate register using $\sim \log_2(k)$ qubits and then verifying using a multi-controlled NOT gate whether the sum equals 1 and recording this on a separate qubit called \textit{flag}. Once the flag is set, we apply rotation in the $Z$ basis proportional to the trained parameter in QAOA. Both the program and the resulting quantum circuit can be found in Fig.~\ref{fig:one-hot}. Even with a basic implementation of adders~\cite{ruiz2017quantum}, and a usually costly no-ancilla MCT implementation~\cite{DaftWullie2022stack} for checking the sum state, this alternative circuit requires no more than $\tildorder{k}$ gates. Notably, this number is smaller than the previously presented term-by-term implementation.

The corresponding new constraint function $f'_C$ differs from $f_C$ and can be explicitly written as  
\begin{equation}
	f'_C(b_1,\dots,b_k) =  \left [ \sum_{i=1}^k b_i \neq 1   \right ],
\end{equation}
where $[\psi]$ is the Iverson notation that evaluates to 1 if $\psi$ is satisfied and 0 otherwise. Note that the explicit form of $f'_C$ nor the corresponding Ising model is not needed to implement the circuit given in Fig.~\ref{fig:one-hot}. Yet, one can show that the corresponding Ising model has exponentially many terms, which with the state-of-the-art approach, will also result in exponentially many gates, see Appendix Sec.~\ref{sec:exp-terms}. As the representation as a pseudo-Boolean polynomial (in this case also Boolean polynomial) is unique, one cannot hope for finding a different polynomial for $f'_C$ with a smaller number of terms. 

One could claim that the above constraint leads to a flat energy landscape. Indeed, one can show that for any valid or invalid solutions, we have $f'_C(b) \in \{0,1\}$, while for the state-of-the-art method, the energy span $f_C(b) \in \{0, 1, 4, \dots, (k-1)^2\}$. While, as we will show later, the former is much more efficient in terms of the number of measurement outcomes required, the latter indicates the extent of deviation from the desired outcomes. Nevertheless, the same can be accomplished by a small update to our program presented in Fig.~\ref{fig:one-hot}. Since in our program, we are computing the sum of all bits \emph{sum}, one can further compute $\text{sum} \mapsto (\text{sum}-1)^2$ using basic operations. Then instead of applying $ Z $ rotation on the flag which is not needed anymore, we can just rotate it as $ \ket{\text{sum}} \mapsto \exp(i \cdot \text{sum}) \ket{\text{sum}}  $ and uncompute to implement exactly the same constraint function $f_C$.  Since $sum$ is defined over $\log(k)$ bits, this will not require more than $\order{\polylog(k)}$ additional operations on quantum hardware. Hence, one can implement $ f_C $ with $\tildorder{k}$ gates. 

Now, we can draw several important conclusions. Firstly, through this toy example, we demonstrated the superiority of our approach over the state-of-the-art implementation in terms of the number of gates required to verify the validity of one-hot encoding. Furthermore, we showed that the pool of constraint functions $f_C$ suitable for efficient implementation on quantum hardware can be expanded. We lift the requirement that the constraints must be represented as pseudo-Boolean polynomials with a polynomial number of terms. We now only need the function to be efficiently implementable on quantum hardware, which is at the same time necessary requirement for the NISQ era. This newfound freedom enables us to choose functions that may offer a better landscape than the ones in consideration. Finally, the overall implementation can be improved by optimizing parts of the circuit that are easy to interpret (like arithmetic operations), taking into account the limitations of the quantum hardware, such as limited connectivity or the ability to apply gates in parallel.

It is worth mentioning that increasing the pool of constraint functions and reducing the number of gates are not the only possible advantages of the proposed approach. Another possible benefit of the program-based approach is the potential for reducing the search space. To illustrate this, let us consider an inequality constraint of the form $|g(x)| \geq k$ for some fixed constant $k>0$ and polynomial function $ g $. The state-of-the-art method is to convert it into two constraints of the form
\begin{gather}
g(x)  + Ny \geq k,\qquad
-g(x)  + N(1-y) \geq k,
\end{gather}
where $y$ is a new binary variable to be optimized, and $N$ is a sufficiently large constant. Then the inequalities are converted into equalities using slack variables, which are then transformed into soft constraints for QUBO, using the penalty method. This results in new binary variables to be optimized for each such constraint, increasing the search space.

Contrary, using the program approach, it is enough to compute $z \coloneqq g(x)$ efficiently (note that previously we had to assume that $g(x)$ is polynomial, which is not needed here). Then, if $z$ is stored in a usual binary format for a signed integer, we can use an implementation like the one described in~\cite{orts2023efficient} to compute $|z|$. Afterward, it is sufficient to verify $|z| \geq k$, which can be accomplished by checking if $|z|-k \geq0$, which is again simple using signed integer representation. All (except for possibly computing $g$) operations involved consist of only basic arithmetic operations, yet allow applying a generic constraint without increasing the search space.

Keeping in mind the above examples, we can outline the following scheme for designing a Program-based QAOA (Prog-QAOA) ansatz. Visualization can be found in Fig.~\ref{fig:diagram}.
\begin{enumerate}
	\item \textit{Problem description}: Provide a precise description of the problem, including the variables, objective value, and constraints.
	\item \textit{Program preparation}: For the objective function $f$, prepare a program that computes the objective value. For each constraint $C_i$, prepare a program that outputs 0 if the constraint is satisfied and a positive number otherwise. Note that the explicit form of the corresponding constraint function $ f_{C_i} $ is not necessary.
	\item \textit{Circuit implementation}: Implement each program one by one as quantum circuits as follows. Start by constructing the quantum circuit that produces the desired outcome of the program on a register $\ket{\rm out}$. Then apply the operation $\ket{\rm out} \mapsto \exp(iA\theta)\ket{\rm out}$ in the computational basis, like e.g. in \cite{gilliam2021grover}, where $A$ is the corresponding penalty parameter for the program and $\theta$ is the optimized parameter. Finally, implement the inverse of the quantum circuit representing the program.
	\item\textit{Mixer and initial state implementation}: Combine the constructed circuits with appropriate input and mixer circuits.  
\end{enumerate}

\begin{figure}[t]
    \centering
    \includegraphics[]{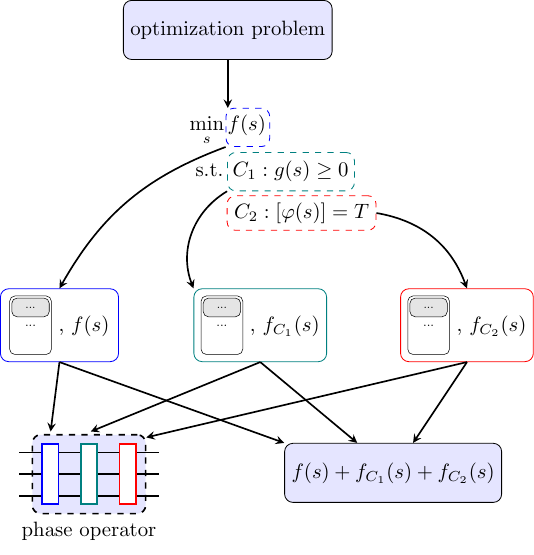}
    \caption{A high-level diagram of Prog-QAOA. For simplicity, penalty terms and parameters $\theta$ for ansatz optimization are omitted.}
    \label{fig:diagram}
\end{figure}

\begin{figure}[t]
\begin{algorithm}[H]
\caption{Program for verifying if bits form a one-hot vector}
		\begin{algorithmic}[]
		\REQUIRE $(b_{0},\dots,b_{k-1})$ -- bits
		\STATE $\text{sum} \gets 0$
		\FOR{$i=0$ to $k-1$ }
		\STATE $\text{sum} \gets \text{sum} + b_{i}$
		\ENDFOR 
		\STATE flag $\gets 0$
		\IF{$\text{sum} \neq 1$}
		\STATE flag $\gets 1$
		\ENDIF
		\RETURN flag
	\end{algorithmic}
\end{algorithm}

\centering
\includegraphics{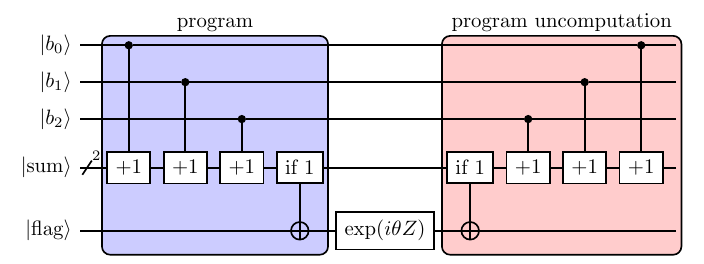}
 \caption{Program for checking if bits form a one-hot vector (top) and the corresponding quantum circuit for $k=3$ (bottom).\label{fig:one-hot}}
\end{figure}

By following this scheme, we create a Prog-QAOA that incorporates both the objective value and constraints into the quantum circuit. One may choose a more advanced mixer and initial state as suggested in \cite{hadfield2019quantum,wang2020x,bartschi2020grover,sawaya2022encoding}. By doing so, it is possible to reduce the set of constraints and further decrease the computational resources, provided that the initial state and mixer are themselves reasonable to implement.

One may wonder how to transform rather abstract programs into a low-level quantum circuit representation. First, since these programs are expected to consist of basic operations such as arithmetic operations and conditionals, we claim that this transformation even if inefficiently should be straightforward to implement by hand using known implementations of these operations. However, one can use a more generic approach by implementing the classical program, compiling it into machine-level language acting on bits, and turning this into a reversible classical circuit which then can be straightforwardly transformed into a quantum circuit. We refer the reader to the relevant literature for this process~\cite{miller2009synthesizing, saeedi2013synthesis}. Which approach would turn out to be more beneficial may depend on the particular program to be implemented.

\section{Applications of Prog-QAOA}\label{sec:applications}

In this section, we show how one can apply Prog-QAOA to paradigmatic and important problems like the Max-$K$-Cut and the Travelling Salesman Problem. In particular, we show that the usual approaches for TSP were doomed to fail in achieving an optimal number of gates due to limitations inherent in the QUBO representation. Finally, we show that, for the two aforementioned problems, our implementations are near-optimal and outperform existing approaches despite the extensive literature dedicated to these problems.

\subsection{Max-$K$-Cut} 

Let $G=(V, E)$ be an undirected complete graph with the weights $W:E\to \RR_{\geq 0}$ and $K>2$ allowed colors. For convenience we will use notation $W_{i,j}$ instead of $W(\{i,j\})$. The goal of Max-$K$-Cut is to color the nodes of $G$ to maximize the total weight of the edges connecting nodes with different colors~\cite{tabi2020quantum}. Alternatively, one may consider minimizing the total weight of the edges connecting nodes with the same colors. Note that there is an associated decision problem called Graph Coloring where the question is `Is there a node coloring such that \emph{all} edges connect nodes with different colors?', and it is defined for unweighted graphs. While these problems are different in nature, they can be addressed in a similar manner.

Let us assume that we are given a list of integers $ c_i $, denoting the color of node $ i $. By looking at the abstract problem definition, we see that we need to calculate the total weight of the edges connecting the nodes with the same colors, which will serve as our objective to minimize. 
An intuitive program for computing the objective value can be found in Algorithm~\ref{alg:maxkcut_control}. To implement the algorithm on a quantum device, each integer is mapped to a quantum register. In our algorithms, we will represent integers using binary encoding. Thus, $\lceil \log K \rceil$ qubits are required for each register. For each pair of different integers $(c_i,c_j)$, we check if they encode the same color and store this in an additional (qu)bit variable `flag' (implementation details of each basic operation can be found in Appendix Sec.~\ref{sec:implementation-details}). Based on the `flag', the value $W_{i,j}$ is added to the `sum' register, initially set to 0. Note that, in general, the addition must be performed using floating-point arithmetics. 

Note that the variables $i$ and $j$ in Algorithm~\ref{alg:maxkcut_control} are at each step in a computational basis state. Hence, any operation performed based on their values can be directly implemented within the quantum circuit, eliminating the need for storing them in separate quantum registers and performing classically-conditioned operations based on them. This allows us not only to free some qubits but avoids the need of storing $W_{i,j}$ say in a QRAM, which would be necessary in case $i$ and $j$ were stored in quantum registers.

\begin{algorithm}[t!]
	\caption{Program for calculating the objective in Max-$K$-Cut problem}\label{alg:maxkcut_control}
	\begin{algorithmic}[]
		\REQUIRE $(c_1,\dots,c_n)$ -- the list of colors associated to each node, $W$ -- weights of the graph
        \STATE $\text{sum} \gets 0$
		\FOR{$i=1$ to $n$ }
		\FOR{$j$ in $[1,2,\dots,i-1,i+1,\dots n]$ }
        \STATE $\text{flag} \gets c_i=c_j$
        \IF{$\text{flag}$}
        \STATE $\text{sum} \gets \text{sum} + W(\{i,j\})$
		\ENDIF\ENDFOR \ENDFOR
        \RETURN sum
	\end{algorithmic}
\end{algorithm}

On the other hand, the `flag' register may be seen as unnecessary as the summation can be performed controlled on quantum registers corresponding to $c_i$ and $c_j$. While this would be possible, it would require performing the summation conditioned on multiple qubits, resulting in a more complicated circuit. Instead, by storing such a check separately, we can implement the summation controlled only on a single bit. This significantly reduces the number of gates due to a smaller number of controlled bits. Note that both scenarios can be precisely described with the presence of the variable `flag' in the program, showing the flexibility of our approach.

Since the problem is inherently unconstrained, providing a program only for the objective function is sufficient. However, a constraint may be needed to ensure that the colors belong to the set $\{0,1,\dots,K-1\}$, for instance, in the case when one starts in the uniform superposition of all bits. In such a case, one may penalize the invalid solutions by reusing the circuits for binary verification described in~\cite{botelho2022error} to penalize the invalid states (in such a case instead of measuring the qubits one can simply apply $Z$ rotations) or use the approach presented in which the values greater than $K-1$ would correspond to valid color say $K-1$~\cite{fuchs2021efficient}. In this paper, we overcome this problem by starting in an equal superposition of only valid states -- implementation details of the circuit in the case of binary encoding can be found in Appendix Sec.~\ref{sec:sim-maxkcut}. 

With this choice of the initial state, we choose the Grover mixer for the proposed Prog-QAOA~\cite{bartschi2020grover}. Note that alternatively, one can apply the Hamiltonian transitioning the amplitude only between consecutive colors or between all colors within a single register, which correspond to ring and complete $XY$-mixers working in the space of one-hot states~\cite{wang2020x}. In the first case, as the Hamiltonian is a circulant matrix, it is enough to implement the generalized Quantum Fourier Transform (QFT), the diagonal matrix consisting of eigenvalues of the circulant matrix, and the inverse of QFT~\cite{marsh2020combinatorial}. While we can use the same procedure for the complete $XY$-mixer, it is enough to just apply the Grover mixer on each quantum register in parallel. 

To show that the simplicity of Prog-QAOA does not come at the expense of increased resources, let us calculate the number of qubits and gates required for our approach. Note that a more comprehensive analysis will be done later after introducing advanced techniques for preparing more efficient circuits. We need $n \lceil log K \rceil=\order{n\log K}$ qubits to represent integers $c_i$, one qubit for the `flag' variable and $\order{ \log (n^2) } = \order{\log(n)}$ qubits for the `sum' (we assume that values $W_{i,j}$ are of $\order{1}$, so the maximal value of the `sum' is at most $n^2$, and we desire constant precision). Hence, we obtain a similar number of qubits as in the state-of-the-art works such as~\cite{fuchs2021efficient,tabi2020quantum}. Regarding the number of gates, we need to perform $\order{n^2}$ checks, with each requiring $\polylog(K)$ gates, and $\order{n^2}$ controlled additions, with each requiring $\order{\polylog(n)}$ gates. As a result, the total number of gates is $\order{n^2 \polylog(n,K)}$, which improves any implementation found in the existing literature~\cite{lucas2014ising,wang2020x,tabi2020quantum,fuchs2021efficient}. Later on, we will demonstrate how this number can be further reduced to $\order{n^2 \polylog(K)}$.

\subsection{Travelling Salesman Problem}

The goal of the Travelling Salesman Problem is to find the path with the minimum cost that passes through all cities and returns to the starting city on the complete graph $K_n=([n], E)$, where $[n]\coloneqq\{0,1\dots,n-1\}$. An instance of the TSP problem is defined by a cost function $W:E\to \mathbb R$. Similarly, as for Max-$K$-Cut, we will write $W_{i,j}$ instead of $W(\{i,j\})$.

Let $c_t$ be an integer denoting the city visited at time $t$ encoded in binary like it was for Max-$K$-Cut, and $\{c_t\}$ be the set of all such integers. Contrary to Max-$K$-Cut, for this problem, we also need to include a constraint $C$ ensuring that each city is visited only once. Let us start with implementing the constraint. We will use $n$ registers where each $c_t$ will be mapped to a quantum register. In~\cite{glos2020space}, this constraint was decomposed into the conjunction of constraints $c_t \neq c_{t'}$ for $t\neq t'$. Alternatively, the initial state can be constructed as the equal superposition of permutations over $[n]$~\cite{bartschi2020grover}, in which case the need for constraints is eliminated. However, both approaches require $\tildorder{n^3}$ gates. 

Instead, we propose a permutation verification program given in Algorithm~\ref{alg:permutation}. The idea of the algorithm is to count occurrences of each city $i\in [n]$ and store the count in a new register called count$_i$. Then $\{c_t\}$ forms a permutation iff all count$_i = 1$. Note that the algorithm also checks whether $c_t<n$ for all $t$ as a side benefit, making it optional to start in a superposition of $\{0,1,\dots,n-1\}$. Instead of counting the exact number of occurrences, it is sufficient to count the parity of the occurrences, which replaces addition modulo $n$ with addition modulo 2 on a single-qubit count$_i$. This can be implemented simply by a multicontrolled-NOT gate, with count$_i$ as the target. This follows from the fact that for any incorrect assignment, at least one city is not visited, and the corresponding $c_i$ is $0$. This simplification saves quite a significant number of qubits and gates, although it does not significantly impact the cost complexity.

Note that such a high-level classical program may require additional resources, which are not evident at first but may be needed during quantum circuit implementation. This is because some basic operations might not be reversible and require auxiliary qubits. 
For instance, assigning 1 to the flag register if `count$_i$ is not equal to 1' is not reversible as the flag does not have a predetermined value. This can be avoided in two ways. One approach is to use a separate auxiliary bit $b_i$ for each of count$_i$, initially set to $1$. Then, if the condition is met, the value of `flag' is simply swapped with the value of $b_i$. We added two comments in Alg.~\ref{alg:permutation} that show how to adapt the program accordingly. 
A better approach is to compute the negation of logical AND operation applied to all count$_i$ values and store the result directly in the flag bit-- this can be easily achieved with a multi-controlled NOT gate (for which we anyway may want to use additional qubits). Note that all these alternatives can be precisely expressed within the framework of the algorithm.

Let us analyze the number of qubits and gates in the aforementioned program. Note that we need $n\lceil \log n\rceil $ qubits for $c_i$, $2n$ qubits for count$_i$ and $b_i$ and 1 qubit for `flag'. Note that variables $i$ and $t$ do not need to be stored on a quantum computer. Regarding gates, we have $\order{n^2}$ controlled additions modulo 2, which can be implemented using only $\order{\polylog n}$ gates as explained in Appendix Sec.~\ref{sec:implementation-details}. In addition, we need $\order{n}$ gates for setting the value of `flag' (since each time the operation acts only on 3 qubits, the number of gates per count$_i$ is constant). Thus, in total, we need $
\order{n \log n}$ qubits and $\tildorder{n^2}$ gates. The corresponding constraint function is 
\begin{equation}
        f_C = [ \{c_t\}\text{ is not a permutation}].
\end{equation}
It could be argued that constraint function with co-domain $\{0,1\}$ is not suitable for the optimization process due to the exponentially decreasing ratio of permutations to all possible assignments, i.e. $\frac{n!}{n^n} = \order{e^{-n}\sqrt{n}}$. However, if count$_i$ would not be computed modulo 2, one instead compute $f_C'=\sum_i(\text{count}_i -1)^2$ or $f_C''=\sum_i|\text{count}_i -1|$ as introduced in the previous subsection, which will require $\tildorder{n}$ qubits and gates. In fact, the first approach with the squared constraint function corresponds to an exact simulation of QAOA with the W-state~\cite{wang2020x} as the input state for the appropriate mixer.

\begin{algorithm}[t]
	\caption{Program for permutation verification}
    \label{alg:permutation}
	\begin{algorithmic}[]
		\REQUIRE $(c_0,\dots,c_{n-1})$ -- the list of integers to be verified if they form a permutation of the range $\{0,1,\dots,n-1 \}$
		\FOR{$i=0$ to
			$n-1$ } 
		\STATE count$_i  \gets 0$  \hfill {// if used, also $b_i=1$}
		\ENDFOR 
		\FOR{$t=0$ to $n-1$} 
		\FOR{$i=0$ to $n-1$} 
		\IF{$c_t = i$ } 
		\STATE count$_i \gets  ($count$_i + 1) \mod 2 $ \hfill {// modulo 2 is optional}
		\ENDIF
		\ENDFOR 
		\ENDFOR 
		\STATE flag $ \gets 0 $
        \STATE {// Alternatively to the loop below `flag $\gets \neg$ AND$_{i=0}^{n-1}\;  {\rm count}_i$'}
		\FOR{$i=0$ to
			$n-1$ }  
		\IF{count$_i \neq 1$ } 
		\STATE ${\rm flag} \gets 1$ \hfill {// more precisely swap `flag' with $b_i$}
		\ENDIF
		\ENDFOR
		\RETURN flag
	\end{algorithmic}
\end{algorithm}

\begin{algorithm}[t]
	\caption{Program for computing the cost of the route.}
	\label{alg:tsp-costroute}
	\begin{algorithmic}[]
		\REQUIRE $(c_0,\dots,c_{n-1})$ -- the list of integers denoting visited cities, $W$ -- cost matrix
		\FOR{$i=0$ to $n-1$} 
		\FOR{$t=0$ to $n-1$} 
		\STATE $\textrm{flag} \gets c_{t}= i$
		\IF{$\textrm{flag} = 1$} 
		\STATE $\textrm{edge}_i \gets c_{t+1}$ 
		\ENDIF 
        \STATE \textbf{uncompute} $\textrm{flag} \gets c_{t}= i$
		\ENDFOR
		\ENDFOR 
		\STATE cost $ \gets 0 $
		\FOR{$i=0$ to $n-1$} 
		\FOR{$j=0$ in $[0,1,\dots, i-1,i+1,\dots,n-1 ]$}
		\IF{$ \textrm{\rm edge}_i =j$}
		\STATE {cost $\gets$ cost $ +W_{i,j}$  }
		\ENDIF
		\ENDFOR 
		\ENDFOR
		\RETURN cost
	\end{algorithmic}
\end{algorithm}
Finally, we should calculate the total cost of the route given the list of visited cities. The usual approach is to consider all pairs of consecutive time points, identify the cities visited at each time point, and then incorporate the appropriate cost~\cite{lucas2014ising,glos2020space,wang2020x,bartschi2020grover}. However, all of these approaches resulted in Ising models with $\tildorder{n^3}$ terms. This is expected because for each pair of consecutively visited cities $(c_t, c_{t+1})$, we need to consider all possible pairs of cities, which requires encoding the \emph{entire matrix $W$}. Since the $W$ matrix has $\order{n^2}$ degrees of freedom, and we have $n$ such pairs of consecutive times,  $\order{n^3}$ gates are expected in total as we will show later. On the other hand, if we examine Miller-Tucker-Zemlin Integer Linear Programming (MTZ ILP) formulation~\cite{miller1960integer}, we can see that by choosing a city-to-city encoding with edge$_{i}$ representing the next visited city after city $i$, the objective value is simply given by $\sum_{i,j} [\text{edge}_i=j] W_{i,j}$. Note that this representation allows us to check each possible combination of consecutively visited cities only once. Unfortunately, as demonstrated in Appendix Sec.~\ref{sec:cost-tsp}, MTZ ILP requires cumbersome and costly checks to determine whether $\{\text{edge}_i\}_i$ forms a permutation, which was simpler with time-to-city encoding.

We propose to keep the original representation for verification of the permutation and, for the sake of computing the objective value, switch to city-to-next-city encoding. The program is given in Algorithm~\ref{alg:tsp-costroute}. To implement the program, we start with the same integers $c_t$. The algorithm starts by changing the time-to-city representation into city-to-next-city. Each step consists of two parts: checking if $c_{t}$ is equal to city $i$ and then copying the value of $c_{t+1}$ into ${\rm edge}_i$, which will hold the city visited after city~$i$. Then, with this new representation, we can implement an algorithm similar to the one used for Max-$K$-Cut to compute route cost. Note that within the same program, we use the procedure \textbf{uncompute} that simply applies the inverse of the marked step (in this case, uncomputes the value of the flag).

Quantum circuit corresponding to the Algorithm~\ref{alg:tsp-costroute} requires
$n\lceil \log n \rceil$ qubits for edge$_i$ integers, 1 qubit for `flag' and $\order{\log n}$ qubits for `cost', assuming $\max_{i,j} W_{i,j} = \order{1}$. Note that integers $i$, $j$, and $t$ are again not required to be stored on the quantum computer. Changing the encoding requires $\order{n^2}$ checks for $c_t=i$ and controlled copies from $c_{t+1}$ into edge$_i$, both of which can be implemented with $\order{\log n}$ gates. Then, we need $\order{n^2}$ checks for edge$_i =j$ and controlled additions, which also requires $\order{\log n}$ gates. Altogether, both Algorithms \ref{alg:tsp-costroute} and \ref{alg:permutation} require $\order{n\log n}$ qubits and  $\tildorder{n^2}$ gates. Using Algorithms \ref{alg:permutation} and \ref{alg:tsp-costroute}, the choice of mixer includes both $X$-mixer or Grover mixer~\cite{bartschi2020grover}, or mixer proposed previously for Max-$K$-Cut.

\subsection{Other problems} 

Prog-QAOA has been demonstrated so far for two important paradigmatic problems. However, its application goes beyond the cases considered here and is applicable even to more complicated problems. To support this claim, we refer the reader to Appendix Sec.~\ref{sec:additional-examples}, where we present how the Prog-QAOA can be implemented for Integer Linear Programs, Travelling Salesman Problem with Time Windows, and a special case of Graphs Isomorphism. The first two problems are known to be NP-complete, while it is not known for the last one whether it is in P (although a quasi-polynomial algorithm is known for general case~\cite{babai2016graph}). For all of these problems, we were able to reach lower bounds for the number of qubits. In addition, the number of gates required is smaller than the ones in the referenced cases.

\section{Technical Details}\label{sec:technical}

\subsection{Parallelization and yield} In the previous section, we provided simple programs for important and paradigmatic optimization problems. Although it may seem straightforward to design quantum circuits for a given program, more efficient circuits can be obtained, taking into account the capabilities of the quantum devices. For instance, certain gate-based quantum hardware, such as IBM quantum computers, allow parallel application of the gates. This suggests that the operations involving different variables can be parallelized and executed simultaneously. For example, let us consider Algorithm~\ref{alg:maxkcut_control}. The checks $c_i=c_j$ can be run in parallel provided that flag$_{i,j}$ exists for each such pair. While by adopting this approach, the number of required qubits would increase into $\order{n^2}$, the depth would decrease from $\tildorder{n^2}$ to $\order{n\log n}$ using round-robin tournament ordering~\cite{rasmussen2008round}. Unfortunately, the `sum' variable cannot be parallelized that easily, as each addition operation `blocks' access to this variable. One could consider, at the cost of extra auxiliary variables, doing a quasi-parallel summation, in which partial sums are computed for consecutive pairs of addends, then partial sums are computed for consecutive pairs of partial sums,  etc.

However, it can be shown that there is no need for $\order{n^2}$ flags in this case, and in fact, there is no need for a `sum' variable. Suppose that the objective function $f$ function can be written as a sum of individual functions:
\begin{equation}
    f_C = f_{1} + \dots + f_{L}.
\end{equation}
Since exponentiation of diagonal matrices is a multiplicative function, implementing a program corresponding to $f$ is equivalent to sequentially implementing programs for each $f_{i}$. The same applies to the constraint functions as well. In the case of Max-$K$-Cut, rather than designing a single program with double loops, we can instead write a program for each pair of different $c_i, c_j$. Then, since we can consider at most $\lceil n/2 \rceil$ pairs with round-robin scheduling in parallel, it is easy to show that only $\order{n}$ flag variables are needed.

Unfortunately, decomposing programs is not always a straightforward task. As an example, let us consider the program for the objective function of TSP given in Algorithm~\ref{alg:permutation}. In this case, one may implement a separate program for each $W_{i,j}$. However, splitting the objective function as was done for Max-$K$-Cut would require repeating the encoding change from time-to-city into city-to-city each time, which is itself resource-consuming. Instead, by recognizing that the objective value is the sum of the weights, one can simply return the value of $W_{i,j}$ each time $W_{i,j}$ whenever it is supposed to be added.

Therefore, in addition to parallel implementation, we introduce the concept of \textbf{yield}. Similar to the \textbf{return} statement, \textbf{yield} returns some value (in the case of the quantum circuit,  it corresponds to applying appropriate $Z$ rotations on the returned `variable'). However, instead of immediately performing the whole uncomputation part, the program continues allowing for the possibility of returning other values. The resulting quantum circuit consists of rotations interleaved with the computation required by other sub-programs, with each rotation appearing in correspondence with the \textbf{yield} statement.  An example program for objective value computation for TSP using the yield concept is presented in Alg.~\ref{alg:tsp-costroute-yield}. 

Several interesting observations can be made here. First of all, it may be seen that our program is independent of matrix $W_{i,j}$. However, with this new formalism, we have effectively incorporated $W_{i,j}$ into $Z$-rotations that are proportional to its value. For clarity, we omitted this and simply added an appropriate comment to the program. For the original Alg.~\ref{alg:tsp-costroute}, $W_{i,j}$ were already incorporated into the variable cost. 

In addition, the uncomputation phase of the whole program is simplified with the usage of \textbf{yield}. Note that since we do not \textbf{yield} or \textbf{return} any value during uncomputation, in the case of Alg.~\ref{alg:tsp-costroute-yield}, the`flag' computation and its uncomputation appear right after each other. These operations clearly cancel out, making the second double-loop unnecessary, so only the uncomputation of the edge$_i$ computation is needed. Through this, we can see that the uncomputation process can be much shorter than the computation itself. 

\begin{algorithm}[t]
	\caption{Program for computing the cost of the route with \textbf{yield}}
	\label{alg:tsp-costroute-yield}
	\begin{algorithmic}[]
		\REQUIRE $(c_0,\dots,c_{n-1})$ -- the list of integers denoting visited cities, $W$ -- cost matrix
		\FOR{$i=0$ to $n-1$} 
		\FOR{$t=0$ to $n-1$} 
		\STATE $\textrm{flag} \gets c_{t}= i$
		\IF{$\textrm{flag} = 1$} 
		\STATE $\textrm{edge}_i \gets c_{t+1}$ 
		\ENDIF 
        \STATE \textbf{uncompute} $\textrm{flag} \gets c_{t}= i$
		\ENDFOR
		\ENDFOR 
		\FOR{$i=0$ to $n-1$} 
		\FOR{$j=0$ in $[0,1,\dots, i-1,i+1,\dots,n-1 ]$}
            \STATE \textrm{flag}  $\gets \textrm{\rm edge}_i =j$
		      \STATE \textbf{yield} flag \hfill //And apply $Z$ rotation proportional to $W_{i,j}$
            \STATE \textbf{uncompute} \textrm{flag}  $\gets \textrm{\rm edge}_i =j$
		\ENDFOR 
		\ENDFOR
	\end{algorithmic}
\end{algorithm}

 \subsection{Limited connectivity circuit optimization}

So far, we have not taken into account the topology of the quantum hardware. Hence, the cost estimations of the algorithms will be valid only for all-to-all connectivity but not for limited connectivity. However, for structured connectivity like the linear nearest neighbor (LNN) or 2D repetitive structures like a 2D grid or a heavy-hexagonal, we can assume that the variables are placed in specific positions of the device. For example, in a grid architecture, a variable requiring four qubits can be stored as a square in the array. In addition to the operations required for the programs, we also need to consider the gates required for swapping qubits~\cite{o2019generalized}.

General hardware architecture connectivity is beyond the scope of this paper but also is not common for large-scale quantum systems -- usually, hardware follows a 1D or 2D pattern. In this paper, we will only focus solely on LNN architecture. In LNN, we assume all variables are arranged in a line. For example, for the case of permutation checking for TSP presented in Alg.~\ref{alg:tsp-costroute-yield}, we can choose an ordering $(c_0,\dots,c_{n-1}, flag, edge_0, \dots,edge_{n-1})$. Then, in addition to the operations, we can introduce `\textbf{swap} var1, var2' commands, which will effectively swap physical (qu)bits corresponding to the logical variables. In general, this kind of swapping requires a different strategy and results in so-called swap networks where each bit is treated as a separate variable~\cite{o2019generalized}. In our case, since we necessarily have variables requiring multi-qubit variables, we will make the necessary adaptations. A more detailed description of the mentioned strategies and cost analysis is provided in Appendix Sec.~\ref{sec:imp-lnn}.

To start with, suppose that we are given an $n$ qubit quantum register where each qubit represents a variable and we want to apply operations between all pairs of variables. In \cite{o2019generalized}, a swap network achieving this task is proposed requiring $n(n-1)/2$ SWAPs. The idea relies on repeatedly swapping pairs of adjacent qubits, ensuring that all pairs are adjacent at some point allowing operation to be applied. In our case, where variables are represented by multiple qubits, it will be desirable to swap entire quantum registers consisting of multiple qubits. Given two registers consisting of $k$ and $m$ qubits respectively, the procedure which we call \emph{register swap} requires $km$ SWAPs and $\sim k+m$ depth. Now, let us consider a sequence of variables ordered as $(v_1,\dots,v_n)$ where each variable is represented by a $k$-qubit register and we want to apply operations on $(v_i,v_j)$ pairs in any order. Combining the strategy for swapping qubits with register swap, this process requires $n(n-1)/2$ register swaps, which in total requires $\sim \frac{1}{2}k^2n^2$ SWAPs and $\sim 2kn$ depth. We call this \emph{intra-all-to-all}. This swap network will be particularly beneficial for Max-$K$-Cut.

Suppose now we have a sequence of variables ordered as $(v_1,\dots,v_n,w_1,\dots,w_n)$ and we want to \textit{interlace} them, i.e. produce a new ordering $(v_1,w_1,\dots,v_n,w_n)$. This may be beneficial, for instance, when one wants to apply operations between $(v_i,w_j)$ pairs. The \textit{interlacing strategy} involves performing bit-by-bit operations on all the qubits, starting from the leftmost qubit of the register storing $w_1$, and swapping this to the left whenever there is a qubit of any register corresponding to $v_i$ on its left side for $i>1$. By following this strategy, we can obtain the desired ordering immediately when all of the bits of $w_i$ are on the right side of bits corresponding to $v_i$, and on the right side of bits corresponding to $v_{i+1}$.  This strategy requires $\sim \frac{1}{2}kmn^2$ SWAPs and $\sim\max(nk+m, nm+k )$ depth on LNN architecture, assuming $v_i$ and $w_j$ require $k$ and $m$ bits respectively.

Another equally important strategy is \emph{inter-all-to-all strategy}, which involves applying operations on $(v_i, w_j)$ pairs when the variables are in the order $(v_1,w_1,\dots,v_n,w_n)$.  Here, we assume that the operation can be applied in any order of such pairs. To achieve this, we fix the registers corresponding to variables $w_i$ and swap the remaining registers repeatedly as in intra-all-to-all, which allows all $(v_i,w_j)$ pairs to be adjacent at some point. In addition to the basic gates required for the operations, $\sim n^2 ((m+k)\min(m,k)+km)$ additional SWAPs and $n(2k+2m+\min(m,k))$ depth are required for this strategy, assuming that $v_i$ and $w_i$ require $k$ and $m$ qubits respectively. This swap network can be used for both permutation checking and, after some adaption, for computing objective value for TSP. 

The strategies mentioned above may be extended to include operations applied on triples or quadruples, but the ones discussed above will be sufficient for our purposes. The details on how we apply the swap networks for Max-$K$-Cut and TSP can be found in the methods.

\section{Cost and Performance Analysis}\label{sec:cost}

In this section, we first identify the most important cost metrics for NISQ-era optimization and then analyze the algorithms we have presented so far to compare them with the existing QAOA variants in the literature. 

\subsection{Cost metrics}

Well-designed circuits are those which use a particularly small amount of resources. Aside from being a multicriterial problem, the magnitude of each criterion may depend on many factors like the properties of the quantum machine or the choice of the classical optimizer. Instead of trying to quantify the quality of the circuit implementation through ambiguous quality measures like noise-robustness, low running time, and the quality of the results, we will consider simple yet important cost metrics that impact the aforementioned measures.

Let us start with the noise-robustness cost metrics. In order to diminish the effect of noise on the computation, it is desired to reduce the hardware computational resources of the quantum circuit. One of the most significant bottlenecks is the number of required physical qubits, so-called the width of the circuit. Reducing the number of qubits allows for the use of smaller quantum hardware, which, in turn, is more likely to yield coherent quantum evolution. Furthermore, each gate introduces noise because of the imprecision of the device, and consequentially the number of gates should be minimized. Finally, the depth of the circuit, typically considered for all-to-all, linear nearest neighbor (LNN), or lattice connectivity of the qubits, should be minimized to reduce the decoherence~\cite{wang2021noise}. One should take into account that minimizing the number of qubits often enlarges the depth of the circuit. This trade-off between the depth and the width was observed for many problems and is of great interest~\cite{glos2020space,fuchs2021efficient,tabi2020quantum}. To combine these two metrics, we will consider the volume of the circuit defined as the product of the depth and width of the circuit.

To minimize the run-time of the optimization, it is also important to minimize the number of required circuit runs. That depends on many factors, including the choice of the optimization algorithm and the quantum error mitigation scheme used. Because of the equivalence between diagonal Hamiltonians and the pseudo-Boolean function, one can consider the Chebyshev bound to be used for estimating the minimum number of shots. However, the Chebyshev bound requires knowing the upper bound on the variance of the sampled variable. Instead, it is much easier to use Hoeffding inequality to show that for precision $\varepsilon$, we need 
$M \geq  \frac{-(E_{\max} - E_{\min})^2 \log (p/2)}{2\varepsilon^2}$,
measurements, where $E_{\max }$ and $E_{\min}$ are the maximal/minimal energies of the pseudo-Boolean function, and $p$ is the probability of Hoeffding inequality to work. Since $E_{\min}$ and $E_{\max}$ are not known, their lower and upper bounds respectively can be used. Note that in such case we can use the energy span of the objective function for determining the lower bound on the number of shots required.

First or higher-order derivatives may be required depending on the optimization procedure used. In fact, recent results using Rydberg atom arrays based quantum hardware suggest that first-order methods like Adam may be of great interest~\cite{ebadi2022quantum}. This is because estimating gradient or higher-order derivatives does not complicate the quantum circuit thanks to the shift rule, but only increases the total number of required measurement samples. While one can use the finite difference method to approximate the gradient, a more common way (and also more precise) is to compute it as a linear combination of the derivatives coming from the trainable parameterized gates. This provides an unbiased estimator of the gradient, which is important for reaching
high-quality optima~\cite{sweke2020stochastic}. One could use the simultaneous perturbations method, which at small costs may provide sufficient estimation, but also one can estimate each addend of the gradient estimator independently as in \cite{sweke2020stochastic}. For the former two methods, energy span is the only cost metric required; however, for the latter method, the number of required measurements is proportional to the number of parameterized gates used in the ansatz, which is our next cost metric. Note that the first two methods may be beneficial against the last one only if many parameterized gates share a single optimized parameter, which is typical for QAOA, but not so usual for VQE. One also has to remember that executions of different unitaries is a resource-consuming task~\cite{zhou2023performance}.

Finally, it is also essential that the classical optimization algorithm can efficiently investigate the landscape of the objective function, thus being more likely to provide high-quality results. This can be attained by limiting the subspace of evolution to a smaller subspace containing the solution, which we call the search space. In the original QAOA proposed by \cite{farhi2014quantum} with the $X$-mixer, the search space is the whole Hilbert space. However, for many problems, a much smaller search space can be obtained with an ingenious formulation and mixer selection~\cite{bartschi2020grover,wang2020x,hadfield2019quantum}. For example, for the TSP over $n$ cities, the dimension of the search space can be dropped from $2^{n^2}$~\cite{lucas2014ising} to $n^n=2^{n\log n}$ with $XY$-mixer~\cite{wang2020x} or even to $n!\sim 2^{n \log n}$ with Grover mixer~\cite{bartschi2020grover}, despite all requiring $n^2$ qubits. This enables faster convergence to high-quality solutions. We propose comparing the logarithm of the dimension of the search space, which is called in this paper the size of the search space. Hence, if the space consists of only polynomially many times more elements, it will not affect the complexity of the search space size. Eventually, the logarithm of the dimension has a natural interpretation of `how many qubits would we need if we could squeeze the search space into the minimal number of qubits'. The search space size is not larger than the number of qubits and not smaller than the logarithm of the number of feasible solutions to the problem.

In the case of the number of gates, it is natural to expect that for QAOA, the lower bound grows not slower than the number of degrees of freedom. This is because all degrees of freedom need to be incorporated into the quantum circuit, and each degree of freedom is expected to contribute with at least one (parameterized) gate. Unfortunately, it is less evident for the depth of the circuit -- in principle, more qubits may allow applying gates in parallel within depth $\order{1}$. Such depth was achieved using a parity encoding~\cite{ender2022modular} at the cost of a very large number of qubits proportional to the number of Pauli terms. While the scenario may be particularly applicable for low-depth high-width quantum hardware, we propose here an alternative approach: if the minimal numbers of qubits $a$ and gates $b$ are achieved, then the depth cannot be lower than $\order{b/a}$, assuming that up to 2-local gates are allowed only. Thus we consider the minimal number of qubits over the minimal number of gates as a lower bound on the depth, both on all-to-all and LNN architectures. 

When comparing two implementations, we say that the first is almost as efficient as the second one with respect to the chosen cost metric if the corresponding values differ only up to a polylogarithmic factor. Furthermore, we say that the implementation is near-optimal according to some metric if it reaches the minimum possible value up to a polylogarithmic factor. Note that we drop `near' if their complexities are the same in $\Theta$ notation.

\subsection{Max-$K$-Cut}

Let us begin with determining the lower bounds on the cost measures. Any solution to the original problem can be written as a function $f: [n]\to[K]$ where $[n]$ denotes the set $\{0,1,\dots,n-1\}$, and there are only $K^n$ of them. We can conclude that the lower bound on the search space size and the number of qubits is $\log(K^n) = \Theta(n\log K)$. Since an instance of the problem is defined by $W$, the minimum number of gates and the parameterized gates is $\Omega(n^2)$. Assuming that the optimal number of qubits is used, it is easy to see that the lower bound for the depth is $\order{n\log K}$. 

So far, many different encodings have been proposed for the Max-$K$-Cut  problem~\cite{lucas2014ising,wang2020x,tabi2020quantum,fuchs2021efficient}. The first one represents the color of a node with a one-hot formulation so that we have $n$ quantum registers, each having $K$ bits~\cite{lucas2014ising}.
Therefore, the proposed QUBO formulation requires $nK$ qubits. Because of  $X$-mixer, defined as the sum of 1-local Pauli-$X$, the evolution takes place in the full Hilbert space so the size of the search space is also $nK$. One can easily reduce the search space by using $XY$-mixer and initializing each register with the $W$-state~\cite{wang2020x}. Yet, this encoding still requires $nK$ physical qubits. Encodings with $\sim n\log K$ qubits were achieved by using binary encoding instead of one-hot encoding, using two dissimilar approaches~\cite{tabi2020quantum,fuchs2021efficient}. In~\cite{tabi2020quantum}, authors provide a formulation using higher-order terms. On the contrary, in~\cite{fuchs2021efficient} the authors propose an alternative technique using the Kronecker delta implementation, which we will refer to as Fuchs-QAOA. This is an implicit example of Prog-QAOA, as the problem Hamiltonian is not implemented through direct implementation of each Pauli term. However, in both HOBO and Fuchs-QAOA, the depth seems to be unnecessarily large. Assuming each parameter is introduced with a small number of gates, one expects that a volume of $\order{n^2\polylog(n,K)}$ should suffice. However, for \cite{tabi2020quantum} the volume is $\tildorder{n^2K}$, while for \cite{fuchs2021efficient} it is $\tildorder{n^2K^2}$ in the worst-case scenario, leaving room for improvement for both of them. 

We propose Prog-QAOA implementation which starts in the superposition of basic states encoding only the valid colors in binary, using Grover-mixer as the mixer~\cite{bartschi2020grover} and the program given in Alg.~\ref{alg:maxkcut_control} but taking into account parallization mentioned at the beginning of Technical Details section. For LNN, we use intra-all-to-all strategy. Further details can be found in Appendix Sec.\ref{sec:swap-networks}. Asymptotic analysis shows that the Prog-QAOA approach is the only one that is near-optimal for all cost measures for sufficiently large $K$, see Table~\ref{table:maxkcut}. The derivations are available in Appendix Sec.~\ref{sec:cost-maxcut}.

In order to analyze the dependency in small-$K$ regime, in Fig. \ref{fig:maxkcut-numerics} we present a numerical analysis of the most important cost metrics on the most promising QAOA variants. The details of the numerical experiment can be found in Appendix Sec.\ref{sec:maxkcut-numerics}

Already for $K\leq60$ colors, we can see the advantage of implementing Prog-QAOA. While keeping the scaling in the number of qubits with HOBO and Fuchs-QAOA, Prog-QAOA has significantly smaller depth both on all-to-all and LNN architectures. Starting from $K=17$ colors, the number of gates becomes even smaller than those needed for X-QAOA and XY-QAOA. The depth on LNN also becomes the lowest, starting at $K=45$ colors.

\begin{table}[h!]\centering
	\setlength{\tabcolsep}{8pt}\small 
	\begin{tabular}{@{}lcccccc@{}}\toprule
		& opt.        & X-QAOA  & XY-QAOA   & HOBO  & Fuchs-QAOA  & 
		Prog-QAOA                                                                                                                                          \\\midrule
		qubits      & $n \log K$  & $nK$   & $n K$     & $n\log K$  & $n\log K$      & $n \log K$ \\
		gates       & $n^2$       & $n^2K$ & $n^2K$    & $n^2K$     & $n^2K^2\log K$ & $n^2\log K$		\\
		depth       & $n /\log K$ & $n$    & $n$       & $nK$       & $nK^2 \log\log K$        & $n\log\log K$ \\
		depth (LNN) & $n/\log K$  & $nK$   & $nK^\ast$ & $nK\log K$ & $nK^2 \log K$ & $n\log K$ \\
		energy span & $n^2$       & $n^2K$ & $n^2$    & $n^2$      & $n^2$        & $n^2$\\
		param. gates & $n^2$       & $n^2K$ & $n^2K$    & $n^2K$     & $n^2$         & $n^2$	\\
		eff. space   & $n\log K$   & $nK$   & $n\log K$ & $n \log K$ & $n\log K$     & $n\log K$ \\ \bottomrule
	\end{tabular}
	\caption{\textbf{The resources required by variants of QAOA-type algorithms for
			the Max-$K$-Cut problem.} The table summarizes the cost measures for X-QAOA~\cite{lucas2014ising},
		XY-QAOA~\cite{wang2020x}, HOBO~\cite{tabi2020quantum},
		Fuchs-QAOA~\cite{fuchs2021efficient} and Prog-QAOA. We used the fact that $K< N$,
		$\|W\|_\infty=\order{1}$ and that all penalty parameters are $\order{1}$ when
		simplifying expressions. The derivations are available in Appendix Sec.~\ref{sec:cost-maxcut}.}
	\label{table:maxkcut}
\end{table}

\begin{figure}[t!]
	\centering
	\includegraphics[scale=0.8]{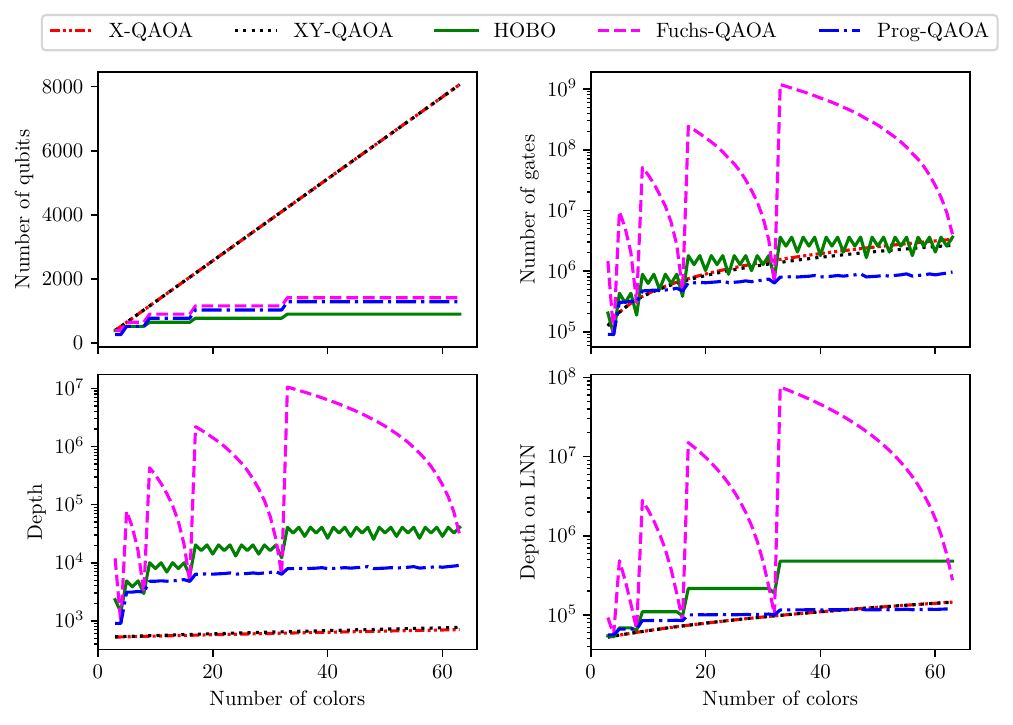}
	\caption{\label{fig:maxkcut-numerics} The scaling of the resources required for various types of QAOA for the Max-$k$-Cut problem. We performed numerical experiments to test the scaling in the number of qubits, number of gates, and depth on both all-to-all and LNN architectures. We used a weighted complete graph with $128$ nodes and colors ranging from $3$ to $63$. Note that the scaling of Prog-QAOA in the number of gates is clearly the best, and it also outperforms other variants in the depth on LNN. Further details on the numerical experiment can be found in Appendix Sec.~\ref{sec:maxkcut-numerics}.}
\end{figure}

\subsection{TSP}

Let us start with estimating lower bounds on the cost measures. There are $n!$ possible routes, which give the lower bound on the search space size $\log(n!)=\Theta(n\log n)$~\cite{glos2020space}. Since any TSP instance is defined by cost matrix $W$, the minimum number of gates is $\Omega(n^2)$, and the same is true for the parameterized gates. Assuming that the problem representation uses the minimum possible number of qubits $\order{n \log n}$, it is easy to show that $\order{n/\log n}$ depth can be reached at the very best. There is a more general argumentation about why we may not be able to go beyond $\order{n}$ depth. If we assume that each qubit represents exactly one city, then it is very natural to expect that it should interact with at least $n-1$ more qubits (representing $n-1$ different cities). Such interactions cannot be done in parallel, indicating that at least $n-1$ depth is needed. Taking these two arguments into account, we find it justified to assume that $\order{n}$ is the optimal possible depth up to a polylogarithmic factor. Finally, let us consider the difference between minimal and maximal energy. The minimal objective value of the original problem can be easily lower bounded by $n\min_{i\neq j}W_{i,j}$; similarly, the maximal energy can be upper bounded by $n\max_{i\neq j}W_{i,j}$. Under the assumption $\|W\|_\infty =\order{1}$, we can assume that the energy span of the Hamiltonian does not grow faster than $\order{n}$.

Many different Hamiltonians have been considered so far for this paradigmatic
problem~\cite{lucas2014ising,gilliam2021grover,miller1960integer,gonzalez2022gps}. The original QUBO formulation presented in~\cite{lucas2014ising} takes the form
\begin{equation}
A_1 \sum_{t={0}}^{{n-1}} \left( 1 -\sum_{i={0}}^{{n-1}} b_{t,i} \right)^{\!\!2} + A_2 \sum_{i={0}}^{{n-1}} \left( 1 -\sum_{t={0}}^{{n-1}} b_{t,i} \right)^{\!\!2} \\
+ \sum_{\substack{i,j={0}\\i\neq j}}^{{n-1}} W_{i,j} \sum_{t={0}}^{{n-1}} b_{t,i}b_{t+1,j},
\end{equation}
where $b_{t,i}$ is a binary variable that is equal to 1 iff city $i$ is visited at time $t$, requiring $\order{n^2}$ qubits. With this QUBO formulation in hand, one can start with the equal superposition state and use the $X$-mixer, which we will refer to as X-QAOA. Alternatively, one may start in a product of $W$-states and use $XY$-mixer~\cite{wang2020x} with an appropriate QUBO without the Hamiltonian part ensuring a single city is visited for each $t$. Furthermore, one may start with the superposition of permutations and use the Grover mixer~\cite{bartschi2020grover} with a QUBO consisting of only the route cost term. We will denote those by XY-QAOA and GM-QAOA, respectively. In \cite{glos2020space}, a HOBO representation is proposed, with $2\lceil \log n\rceil$-local terms, but using only $\Theta(n\log n)$ qubits. In all of the encodings above, a function $f: t \mapsto v$ is encoded, with an interpretation `city $f(t)$ is visited at time $t$'. Finally, we consider Miller-Tucker-Zemlin Integer Linear Programming (MTZ ILP) \cite{miller1960integer}, which can be transformed into QUBO in a usual way~\cite{salehi2022unconstrained}. 
Starting in an equal superposition of valid cities only, using the Grover-mixer as the mixer and using Alg.~\ref{alg:tsp-costroute-yield} and Alg.~\ref{alg:permutation} to implement the phase operator, we obtain Prog-QAOA for TSP that reaches near-optimal values for all cost measures. For both programs, we used inter-all-to-all strategy for LNN. For all-to-all connectivity for permutation verification, we also go over all possible $(t,i)$ using round-robin tournament. Besides our implementation being the best of all the considered, it also breaks the $\order{ n^3}$ volume bound observed in other models.

In Table~\ref{table:tsp}, we present the quality metrics for all the formulations discussed above. The derivations are available in Appendix Sec.~\ref{sec:cost-tsp}. We can see that the circuit volume is $\order{n^3}$ for all of them. For X-QAOA, XY-QAOA, GM-QAOA, and HOBO, it can be shown that the mentioned volume is optimal. All of the formulations encode function $f$ mapping time points $t$ to cities, as explained before. The values $f(t)$ for different $t$ are stored in disjoint quantum registers. However, for each pair $t,t+1$, we have to include the complete information about the cost function $W$ that has $\order{n^2}$ degrees of freedom. If each register has $\order{d}$ qubits, one cannot implement it with a depth smaller than $\order{n^2}$, which gives the final volume $\order{n^2/d \cdot nd}=\order{n^3}$. We would like to note that the results presented in the table for Prog-QAOA with Grover mixer are valid for the other mixers as well.

Furthermore, the LNN architecture has a destructive impact on the depth of X-QAOA and XY-QAOA. In~\cite{o2019generalized}, the authors show that any QUBO can be implemented in depth proportional to the number of qubits it uses. For the QUBO representation from~\cite{lucas2014ising}, we can prove that any swap strategy with any initial arrangement of qubits results in $\Omega(n^2)$ depth using Theorem~\ref{theorem:depth-lowerbound-lnn} presented in Appendix Sec.~\ref{sec:depth-lnn}; hence it cannot be improved.  

\begin{table*}[t!]\centering
	\setlength{\tabcolsep}{6pt}\small 
	\begin{tabular}{@{}lccccccc@{}}\toprule
		& optimal      & X-QAOA & XY-QAOA  & GM-QAOA   & MTZ ILP         &
		HOBO        & Prog-QAOA                                                   \\
		\midrule
		qubits      & $n\log n$ & $n^2$  & $n^2$     & $n^2$     & $n^2\log n$ &
		$n\log n$   & $n\log n$                                                  \\
		gates       & $n^2$     & $n^3$  & $n^3$     & $n^3$     & $n^3$       &
		$n^3$       & $n^2$                                                      \\
		depth       & $n$       & $n$    & $n$       & $n^2\ast$     & $n\log n$         &
		$n^2$       & $n \log n $                                                        \\
		depth (LNN) & $n$       & $n^2$  & $n^2 $    & $n^2\ast $ & $n^2\log n$        &
		$n^2\log n$       & $n \log n $                                                        \\
		energy span & $n$       & $n^3$  & $n^2$     & $n$       & $n^4$       &
		$n^2$       & $n$                                                        \\
		param.gates & $n^2$     & $n^3$  & $n^3$     & $n^3$     & $ n^3$      &
		$n^3$       & $n^2$                                                      \\
		eff.space   & $n\log n$ & $n^2$  & $n\log n$ & $n\log n$ & $n^2\log n$       &
		$n\log n$   & $n\log n$                                                  \\
		\bottomrule
	\end{tabular}
	\caption{\textbf{The resources required by variants of QAOA-type algorithms for the TSP
	problem.}
	The table summarizes the cost measures for X-QAOA~\cite{lucas2014ising},
	XY-QAOA~\cite{wang2020x}, GM-QAOA~\cite{bartschi2020grover}, MTZ ILP~\cite{miller1960integer}, HOBO~\cite{glos2020space}, and Prog-QAOA. We used the fact that $\|W\|_\infty=\order{1}$. `$\ast$' indicates that our estimation may not be tight, and could be theoretically improved. The derivations are available in Appendix Sec.~\ref{sec:cost-tsp}.
 }\label{table:tsp}
\end{table*}

\begin{figure}[th!]\centering
	\includegraphics[scale=0.85]{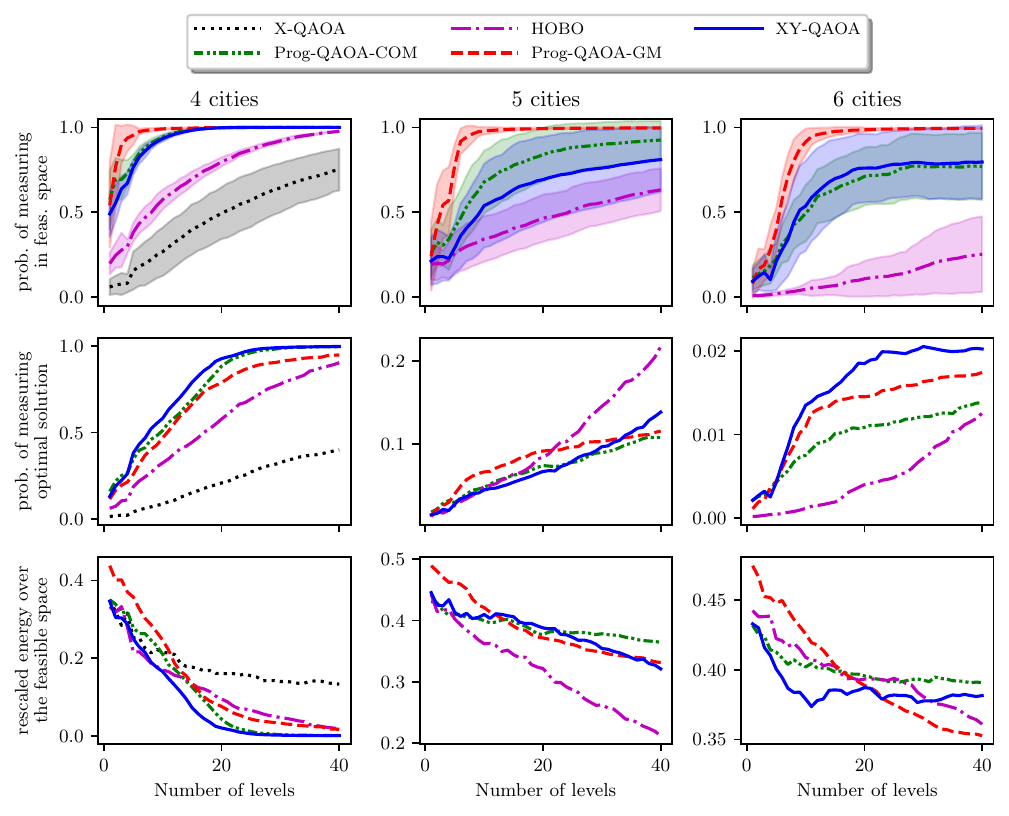}
	\caption{\label{fig:tsp-numerics} Optimization results for various encodings.
		Starting from the top row, we considered the probability of getting a feasible
		solution (a permutation), the probability of obtaining the optimal solution,
		and the mean energy of the optimized quantum state. Note that all the QAOA
		variants exhibit similar quality, with Prog-QAOA-GM and XY-QAOA being the best
		among them. The area in the first row denotes the deviation with respect to the
		standard deviation; for the second and third rows, the areas were much more
		robust and thus were omitted for clarity.}
\end{figure}

In Fig.~\ref{fig:tsp-numerics} we present the optimization results for TSP instances with 4, 5, and 6 cities using different variants of QAOA. We considered Prog-QAOA with Grover mixer (Prog-QAOA-GM) and Prog-QAOA with $b_t$-wise
complete mixer (Prog-QAOA-COM). 
For comparison, we chose the HOBO-based QAOA from~\cite{glos2020space}, X-QAOA \cite{lucas2014ising}, GM-QAOA~\cite{gilliam2021grover}, and XY-QAOA with simultaneous complete $XY$-mixer \cite{wang2020x}. Since the penalty applied strongly depends on the model chosen, to assess the performance, we choose metrics that do not depend on the constraint Hamiltonians: the
probability of measuring a feasible solution, the probability of measuring the optimal solution, and the rescaled energy of the feasible space. The last one is computed by first projecting the quantum state onto the feasible space, then normalizing it, and finally computing its energy against the Hamiltonian $\frac{H-E_{\min}}{E_{\max}-E_{\min}}$, where $H$ is the problem Hamiltonian of the corresponding QAOA, and $E_{\min}$, $E_{\max}$ are its minimum and maximum energies. The details of the numerical experiment can be found in Appendix Sec.\ref{sec:tsp-numeric}.

The only QAOA version that was clearly superior over Prog-QAOA and other presented algorithms was GM-QAOA, with the initial state being the superposition of permutations~\cite{bartschi2020grover} -- the significant difference forced us to remove it from the plots to improve the visibility of the plot. The reason behind this is likely to be a much smaller effective space: for GM-QAOA, the number of quantum states from the pure evolution is only $5! = 120$, while for Prog-QAOA it is $5^5=3125$. Note that the difference in the number of feasible solutions will continue to grow as $\log(n^n)-\log(n!) = \Theta(n)$ based on the Stirling approximation. However, we would also like to note that GM-QAOA is much more costly in terms of both qubits and depth, as presented in Tab.~\ref{table:tsp}.

Yet, using techniques from Prog-QAOA, GM-QAOA can be simulated using only $\order{ n \log n} $ qubits without increasing the depth. It is enough to replace the quantum circuit preparing the initial state for Prog-QAOA with the circuit preparing an equal superposition of permutations where each city is encoded in binary. This can be performed by adapting the procedure for one-hot encodings~\cite{bartschi2020grover}, and the details can be found in Appendix Sec.~\ref{sec:permutations}. With this choice of the initial state, the objective Hamiltonian takes the form 
\begin{equation*} 
\sum_{t=0}^{n-1}\sum_{\substack{i,j=0\\i\neq j}}^{n-1} W_{i,j}\delta(b_t, i) \delta(b_{t+1}, j),
\end{equation*} which can be implemented using Algorithm~\ref{alg:tsp-costroute}. Overall, the depth and the number of gates are equal to that of the GM-QAOA because of the costly initial state preparation, but the number of required qubits is $\order{n \log n}$ instead of $\order{n^2}$. Hence, we can simulate the exact behavior of GM-QAOA using a much smaller amount of qubit resources. Further improvement of the cost measures depends on more efficient initial state preparation. However, at the moment of writing, we are not aware of a more efficient implementation of superposition of permutations neither for one-hot nor for binary encoding.

\section{Discussion and Conclusion}\label{sec:discussion}

Implementing a cost Hamiltonian by applying each Pauli term independently may result in redundancy. Consequently, the cost of running it on a quantum computer increases, and this may detriment the fidelity of the obtained results and hence, the likelihood of quantum advantage. We developed the novel framework Prog-QAOA, built upon the idea of simulating the problem Hamiltonians with simple-to-design quantum circuits using the classical pseudocode function as an input. We showed that Prog-QAOA effectuates a great improvement for Max-$K$-Cut and TSP problems on all the identified quality measures and provides an alternative approach for integer linear programs. It is worth mentioning that the underlying functions still correspond to some Hamiltonians, however, a natural input in the case of a simulated variational algorithm is a pseudo-code, which is much simpler to design than the strict unconstrained mathematical model consisting only of binary variables. This may require further investigation concerning the efficient quantization of classical algorithms. We also would like to point out the strong connection between the oracle design for Grover-type algorithms and the quantum circuit construction for Prog-QAOA.

Proper design of Prog-QAOA results in an efficient or even optimal quantum circuit design, minimizing all cost measures simultaneously. We showed that this effect occurs for important NP-hard problems like Max-$K$-Cut and TSP. For others, we were able to obtain at least a trade-off between the depth and number of qubits, which are commonly considered to be the most significant factors. The importance of such trade-offs has been already pointed out in the literature, as the difficulty of increasing the evolution fidelity or the size of the quantum memory may heavily depend on the chosen technology~\cite{glos2020space,sawaya2020resource}.

Since our approach has a smaller depth, a smaller number of qubits, and a smaller number of gates, it is natural to expect that the fidelity of the resulting quantum states would be much better than any other designs. As many qubits are reused, one can also apply mid-circuit postselection error mitigation techniques by measuring ancilla qubits which should be set to $\ket{0}$ and project the quantum state into the feasible subspace \cite{botelho2022error,mcardle2019error}. Note that the circuits designed for constraint verification, like the permutation verification algorithm introduced in this paper, can be also utilized with such methods. However, one should also mention that the noise in our encoding will have a very different effect than when QUBO is implemented in the usual way because instead of having only a local effect on the quantum state, one can expect here a `global' effect, especially when implementing relatively complicated circuits. We believe this may be an interesting next direction of research, yet to pursue this (at least for TSP), one would need to compare the algorithm with a large quantum computer.

Let us mention that currently QAOA and Quantum Annealing are the two prominent and competing methods for solving combinatorial optimization with quantum hardware. While quantum annealers have currently a remarkable number of qubits (over 5600 by D-Wave) compared to the largest gate-based quantum computer (289 qubits by QuEra Computing Inc.~\cite{ebadi2022quantum}), there is also a significant difference in the numbers of qubits required by the two methods. As an example, let us again consider TSP. As it was shown, for the gate-based model, it is enough to use $\order{ n^2 }$ qubits~\cite{bartschi2020grover,lucas2014ising}, or even $\order{ n \log n}$ qubits as presented in~\cite{glos2020space} or in this paper. However, in the case of existing quantum annealers, one should also take into account the limitation of only 2-local interactions on a limited 2D graph topology. Assuming that the interactions for future quantum annealers will still follow 2D or even 3D topology, the physical interaction graph will have a bounded degree. This means that the number of interactions available on the machine will grow proportionally to the number of physical qubits. However, since for the TSP, we need $\order{n^3}$ interactions for a QUBO~\cite{lucas2014ising,glos2020space}, the embedding will require at least $\order{n^3}$ physical qubits, and thus the quantum annealers need to grow much faster than the gate-based quantum computers. Note that the same problem will occur for any optimization problem defined over permutations. A possible remedy might be to allow more than 2-local interactions between qubits.

Our results could be expanded in several vital ways. First of all, it would be interesting to provide similar circuits for other problems to investigate the limitations of the proposed approach. In particular, we believe the Quadratic Assignment Problem may be an interesting next problem to consider. Our preliminary analysis suggests that while providing the trade-off between the number of qubits and depth is not challenging with the tools developed in this paper, optimizing both cost metrics simultaneously requires a more elaborate investigation. After an in-depth investigation of other optimization problems, one may consider designing a quantum-classical programming language, which will be able to compile the given pseudocodes into efficient quantum circuits. Such software not only would allow simple generation of near-optimal, or at least efficient quantum circuits, but it also would greatly reduce the difficulty in implementing real-world problems into a quantum computer. Taking into account that a deep understanding of low-level details of quantum computing may impose an entry barrier for researchers from different domains, a high-level abstraction would facilitate the engagement of practitioners in the field. Eventually, this will lead to a need for a good overview of implementations of various primitives, either of logical operations (AND, OR), arithmetic operations (addition, multiplication, etc.) or more advanced data structures oriented primitives (database searching, sorting). We proposed and designed some of the primitives like equality, AND (multi-controlled Toffoli),  controlled copy, or arithmetic operations. At the same time it is clear  that such a list is far from being complete, especially when considering compilation for limited connectivity, or when combining those primitives by SWAP networks like ones introduced in the paper. Standarizing and gathering the most widely applicable primitives would be useful for easy and effective design of Prog-QAOA circuits. Finally, one could analyze the quality of the developed framework on existing or near-future quantum hardware to assess its noise-robustness.

\paragraph{Data Availability} The datasets generated during and/or analyzed
during the current study are available in the Zenodo~\url{doi.org/10.5281/zenodo.6867471}.

\paragraph{Code Availability} Code and its description used to generate the
the dataset is available on Zenodo~\url{doi.org/10.5281/zenodo.6867471}.


\paragraph{Acknowledgements}

A.G. has been partially supported by National Science Center under grant agreements 2019/33/B/ST6/02011, and 2020/37/N/ST6/02220. \"O.S. acknowledges support from National Science Center under grant agreement 2019/33/B/ST6/02011. B.B. and Z.Z. acknowledges support from NKFIH through the Quantum Information National Laboratory of Hungary, the Thematic Excellence Programme TKP2020-NKA-06, the QuantERA II project HQCC-10101773 and Grants No. K124152, FK135220, KH129601 and KDP-2023.

\bibliographystyle{plainnat}
\bibliography{optimal-tsp}

\paragraph{Competing Interests} The authors declare that there are no competing
interests.

\paragraph{Author Contribution} AG proposed the idea of the research. AG, \"OS and ZZ selected the relevant optimization problems. AG and \"OS designed the presented algorithms and made an analytical investigation of the presented and referenced algorithms. BB was responsible for numerical investigations for Max-$K$-Cut, and AG was responsible for numerical investigations for TSP. ZZ quantified the impact of the results. All authors contributed to the preparation of the paper.  

\clearpage
\appendix
\onecolumn

\section{Exponential number of terms for one-hot constraint function} \label{sec:exp-terms}
In this section, we will show that the cost Hamiltonian corresponding to HOBO
\begin{equation}
\left [\sum_{i=1}^k b_i \neq 1 \right ] 
\end{equation}
has an exponential number of terms. First, it is straightforward to check that 
\begin{equation}
    \left [\sum_{i=1}^k b_i \neq 1 \right ] 
 =  1 -  \sum_{i=1}^k b_i\prod_{j\neq i} (1-b_j).
\end{equation}
The RHS of above can be transformed into the cost Hamiltonian the usual way by substituting $b_i \leftarrow \frac{1-s_i}{2}$ which gives 
\begin{equation}
    1 -  \sum_{i=1}^k b_i\prod_{j\neq i} (1-b_j) \equiv 1 - \sum_{i=1}^k \frac{1-s_i}{2}\prod_{j\neq i} \left (1-\frac{1-s_j}{2} \right) = 1 - \frac{1}{2^k} \sum_{i=1}^k(1-s_i)\prod_{j\neq i}  (1+s_j).
\end{equation}
Let us compute the coefficient for term $s_{i_1}\dots s_{i_m}$ for $0<m\leq k$. The corresponding coefficient will be
\begin{equation}
    \frac{1}{2^k} ( m \cdot (-1) + (k-m)\cdot 1) = \frac{k-2m}{2^k}.
\end{equation}
We can see that the coefficient equals zero iff $m = \frac{k}{2}$, which is impossible for odd $k$ and happens only for $\binom{k}{k/2} \sim \frac{2^{k}}{\sqrt{k\pi/2}}$ among all $2^k$ coefficients. This shows that the number of coefficients is indeed exponential. Note that the same results can be obtained for the other substitution $b_i \leftarrow \frac{1+\tilde s_i}{2}$ as $\tilde s_i = - s_i$ thus it only impacts the sign of the coefficient.

\section{Depth on LNN for QAOA and XY-QAOA on LNN}\label{sec:depth-lnn}

In~\cite{glos2020space}, authors showed that it is not possible to run QAOA for TSP for $n$ vertices with $X$-mixer on LNN in depth $\order{ n ^{2-\varepsilon}}$ for any $\varepsilon>0$. Here we generalize their reasoning into a general interaction graph $G$ created based on a QUBO. We assume the qubit-qubit nonzero interactions for a QUBO are given by some simple undirected graph $G$, and it has to be implemented on the quantum hardware with LNN qubit connectivity. To implement such interactions, one needs to specify the initial arrangement of the qubits and the swap strategy so that all spins that should interact are neighboring at some point.

\begin{theorem}\label{theorem:depth-lowerbound-lnn} Let $G=(V,E)$ be a simple undirected graph with radius $r$ representing the interactions of a QUBO. For any combination of
	swap networks and initial arrangement of the vertices onto LNN to implement the
	$G$-generated QUBO on LNN architecture, the depth is at least $\Theta(\#V/r)$.
\end{theorem}
\begin{proof}
	Let $n=\#V$ and let $A:V\to [n]$ be an arbitrary bijection of vertices on the
	LNN. Let $\bar v$ be the vertex mapped to the first qubit on the LNN network. Let
	there be a swap network such that it will be
	possible to implement the QUBO in depth $D$. Then it means that all neighbors
	of $\bar v$ in $G$ are within distance $2D+1$ in LNN.
	
	However the same is true for all neighbours of $\bar v$ in $G$, which means that
	all neighbours of neighbours of $\bar v$ in $G$ are at most at distance $2\cdot
	(2D+1)$ from $\bar v$ in LNN. The same argument can be repeated $r$ times. Since
	the radius of the graph is $r$, then it means that all nodes are within distance
	$r(2D+1)$ in LNN. Since the largest distance in LNN is exactly $n-1\leq
	r(2D+1)$, it follows that $D \geq \frac{n-1}{2r} - \frac{1}{2}$.
\end{proof}

Using this theorem, the results immediately follow for any general graph with finite radius $r$, and one can show that the depth is $\Omega(n)$, which also recovers the results presented in~\cite{glos2020space}. It is also worth mentioning that $\order{n}$ depth via the swap networks for general QUBO is presented in~\cite{o2019generalized}.

\begin{theorem}\label{theorem:depth-lowerbound-lnn-constant-radius} Let $G=(V,E)$ be a simple undirected graph with bounded radius representing the interactions of a QUBO. It can be implemented on LNN architecture with minimum depth $\order{\#V}$, which is achievable.
\end{theorem}
\begin{proof}
	One cannot go below $\order{\#V}$ because of Theorem~\ref{theorem:depth-lowerbound-lnn}. On the other hand swap network proposed in~\cite{o2019generalized} allows achieving $\order{\#V}$ for any graph $G$.
\end{proof}

\section{Implementation details for Prog-QAOA} \label{sec:implementation-details}

In this section, we provide implementation details of circuit parts and swap networks used in the paper.

\subsection{Implementation of `is $\ket x$ equal to $\ket k$' for known $k$}
\label{sec:xor} First, we change the quantum state $\ket{x }$ into $\ket{x \oplus k}$, where $\oplus$ is the bitwise addition modulo 2. This can be done with NOT gates applied on the appropriate qubits. Then if $\ket{x}=\ket k$, we will get $\ket 0$ on all qubits. After this, we negate all qubits and apply a NOT operation with controls on all $x$ qubits and target on the ancilla qubit. The circuit takes the form
\begin{center}
	\begin{quantikz}[row sep=1ex]
		\lstick{$\ket{x_{0}}$} &
		\gate[4]{\oplus k} & \targ{} & \ctrl{1} & \qw \\
		\lstick{$\ket{x_{1}}$} &   & \targ{} & \ctrl{1} & \qw \\
		\lstick{$\ket{x_{2}}$} &   & \targ{} & \ctrl{1} & \qw \\
		\lstick{$\ket{x_{3}}$} &  & \targ{} & \ctrl{1} & \qw \\
		\lstick{$\ket 0$} & \qw & \qw & \targ{} & \qw \\
	\end{quantikz}
\end{center} 
If we use the implementation of ancilla-free multi-controlled NOT from~\cite{saeedi2013linear,da2022linear} designed for LNN architecture, we obtain a circuit with depth $\Theta(\log |x|)$, where $|x|$ is the number of qubits of the quantum state $\ket x$.

\subsection{Implementation of controlled copy}\label{sec:controlled-copy}

Let us assume we have three quantum states: source state $\ket{x_s}$, the output state $\ket{x_o}$, and the single-qubit control state $\ket{f}$. We assume that $\ket{x_o}$ is initialized to $\ket{0\cdots0}$. Our goal is to implement the unitary operation
\begin{gather}
\ket{f=0}\ket{x_s}\ket{0\cdots0} \mapsto \ket{0}\ket{x_s}\ket{0\cdots0}, \\ \ket{f=1}\ket{x_s}\ket{0\cdots0} \mapsto \ket{1}\ket{x_s}\ket{x_s} ,
\end{gather}
that copies state $\ket {x_s}$ to $\ket{x_o}$ iff $\ket f = \ket 1$. This can be
implemented in all-to-all architecture as
\begin{center}
	\begin{quantikz}[row sep=1em, column sep=1em]
		\lstick{$\ket{f}$} & \ctrl{1} & \ctrl{2} &
		\ctrl{3} & \qw  \\
		\lstick{$\ket{x_{s,1}}$} & \ctrl{3} & \qw & \qw & \qw \\
		\lstick{$\ket{x_{s,2}}$} & \qw & \ctrl{3} & \qw & \qw \\
		\lstick{$\ket{x_{s,3}}$} & \qw & \qw & \ctrl{3} & \qw \\
		\lstick{$\ket{x_{o,1}}$} & \targ{} & \qw & \qw & \qw \\
		\lstick{$\ket{x_{o,2}}$} & \qw & \targ{} & \qw & \qw  \\
		\lstick{$\ket{x_{o,3}}$} & \qw & \qw & \targ{} & \qw  \\
	\end{quantikz}
\end{center}

Let us now consider the implementation on LNN. First, we apply a strategy similar to interlacing and reorder the registers as 
\begin{equation} \ket{x_{s,1}}\cdots\ket{x_{s,n}}\ket{f}
\ket{x_{o,1}}\cdots\ket{x_{o,n}} \mapsto \ket{f}\ket{x_{s,1}}\ket{x_{o,1}}
\cdots\ket{x_{s,n}}\ket{x_{o,n}}.
\end{equation}
This will require $\order{n^2}$ CNOT gates and $\order{n}$ depth.

Then we implement the following qubit by qubit: 
\begin{center}
	\begin{quantikz}[row sep=1.5ex, column sep= 1.5ex]
		\lstick{$\ket f$} & \swap{1}	& \qw &	\ctrl{1} & \qw & \rstick{$\ket{x_s}$}\\
		\lstick{$\ket {x_s}$}  & \targX{} & \swap{1} &  \targ{} & \qw & \rstick{$\ket{x_o=f \otimes x_s}$} \\
		\lstick{$\ket{x_o=0}$} & \qw & \targX{} &  \ctrl{-1} & \qw  & \rstick{$\ket{f}$}
	\end{quantikz}
\end{center}
This requires $\order{n}$ gates and $\order{n}$ depth.

We end up with the quantum state with controlled copy applied, but in incorrect order:
\begin{equation}
\ket{x_{s,1}}\ket{x_{o,1}} \cdots\ket{x_{s,n}}\ket{x_{o,n}}\ket{f}.
\end{equation}
Now we apply interlacing technique on $\ket{x_s}$, $\ket{x_o}$ registers, which gives us
\begin{equation}
\ket{x_{s,1}}\cdots\ket{x_{s,n}} \ket{x_{o,1}}\cdots\ket{x_{o,n}} \ket{f},
\end{equation}
and finally we move $\ket f$ so that we end up with
\begin{equation}
\ket{f}\ket{x_{s,1}}\cdots\ket{x_{s,n}} \ket{x_{o,1}}\cdots\ket{x_{o,n}} ,
\end{equation}
which is the original order. The number of gates required for this reordering is $\order{n^2}$, and the depth is $\order{n}$. 

Overall, the total number of gates for the controlled copy is $\order{n^2}$ and the depth is $\order{n}$.

\section{Implementation details for Prog-QAOA on limited connectivity}

\subsection{Swap networks}\label{sec:swap-networks}
In this section, we present and analyze the cost of the swap network used in the paper.

Suppose we have $n$ quantum registers, each consisting of a single qubit, and we want to perform some 2-qubit operations on all pairs of registers. In \cite{o2019generalized}, a swap network is provided requiring $n(n-1)/2$ SWAPs. We formally state the result in the following lemma and theorem. Let us remark that at the end of this procedure, the registers end up in the reversed order.

\begin{lemma}[\cite{o2019generalized}] \label{lemma:qubit-all-to-all}
Let $\pi=\pi_2 \circ \pi_1$ be a permutation
	over $[n]$ where
	\begin{equation}
	\pi_1 = (1,2)(3,4)\dots(2k-1,2k), \qquad
	\pi_2 = (2,3)(4,5)\dots(2k,2k+1)
	\end{equation}
	for odd $n=2k+1$, and for even $n=2k$
	\begin{equation}
	\pi_1 = (1,2)(3,4)\dots(2k-1,2k), \qquad
	\pi_2 = (2,3)(4,5)\dots(2k-2,2k-1).
	\end{equation}
	Then the following is true: after applying $\lfloor \frac{n}{2} \rfloor$ times permutation $ \pi $ and then for odd $n$ only one $\pi_1$ on the sequence
	$(1,\dots,n)$, each pair of numbers $i$ and $j$ will appear next to each other at least once.
\end{lemma}
\begin{theorem}[\cite{o2019generalized}] \label{theorem:qubit-all-to-all}
Let $\Hl_{A_1}\cdots \Hl_{A_n}$ be registers, with $\Hl_{A_i}$
	having $1$ qubit each. Let $U_{A_i,A_j}$ be commuting operations 		
	on registers $A_i,A_j$ for $i,j
	\in [n], i\neq j$. Then additional $\sim \frac{1}{2}n(n-1)$ SWAPs and $\sim n$ depth are sufficient to implement the operations $U_{A_i,A_j}$ for $i,j \in [n], i \neq j$ on LNN.
\end{theorem}

In general, one may need to perform operations on every pair of multi-qubit registers. Before generalizing Thm.~\ref{theorem:qubit-all-to-all} for multiqubit registers, let us consider the cost of swapping two registers.

\begin{theorem}[Register swap]\label{theorem:regswap}
	Let $\Hl_A \otimes \Hl_B$ be the product of quantum registers with $k,m$ qubits
	consecutively. There exists a circuit with $km$ SWAPs and $\sim k+m$
	depth which reorders the quantum registers as $\Hl_B \otimes \Hl_A$ on LNN
	architecture.
\end{theorem}
\begin{proof}
	Given the basic state
	\begin{equation}
	\ket{a_1}\cdots \ket{a_k}\ket{b_1}\cdots \ket{b_m},
	\end{equation}
	we apply a sequence of SWAP gates on the qubits $(k,k+1)$, $(k-1,k)$, \ldots,$(1,2)$ to get the state
	\begin{equation}
	\ket{b_1}\ket{a_1}\cdots \ket{a_k}\ket{b_2}\cdots \ket{b_m}.
	\end{equation}
	Note that the procedure requires $k$ SWAP gates. We repeat the procedure for $b_2,\dots,b_k$ each requiring $k$ swaps, which gives $km$ SWAPs in total.
	
	In an `ideal' unrealistic scenario, we would achieve depth $k$ if every SWAP would be performed in parallel. However, a delay is required in order to create an extra qubit between $b_i$ and $b_{i+1}$. For example, consider the following:
	\begin{equation}
	\begin{split}
	\ket{a_1}\cdots \ket{a_k}\ket{b_1}\cdots \ket{b_m} &\mapsto	\ket{a_1}\cdots \ket{a_{k-1}}\ket{b_1}\ket{a_k}\ket{b_2}\ket{b_3}\cdots \ket{b_m}\\
	&\mapsto \ket{a_1}\cdots \ket{b_1} \ket{a_{k-1}}\ket{b_2}\ket{a_k}\ket{b_3}\ket{b_4}\cdots \ket{b_m} \mapsto \cdots.
	\end{split}
	\end{equation}
	
	Let $k\geq m$. It requires $k$ layers of SWAP to create the state
	\begin{equation}
	\ket{b_1}\ket{a_1}\ket{b_2}\ket{a_2}\cdots \ket{b_m}\ket{a_m}\ket{a_{m+1}} \cdots \ket{a_{k}}.
	\end{equation}
	It requires extra $m-1$ SWAP layers to create the desired state, as after each layer, a single $\ket{b_i}$ is in the right position. So for $k\geq m$ the depth is $k+m-1$. Note that $k<m$ is a symmetric case.
\end{proof}

Now we are ready to generalize the swap networks given for single qubit registers in Thm.~\ref{theorem:qubit-all-to-all} to multiqubit registers.

\begin{theorem}[Intra-all-to-all strategy]\label{theorem:intra-all-to-all} Let
	$\Hl_{A_1}\cdots \Hl_{A_n}$ be registers, with $\Hl_{A_i}$
	having $k$ qubits each. Let $U_{A_i,A_j}$ be commuting operations 		
	on registers $A_i,A_j$ for $i,j
	\in [n], i\neq j$. Then additional $\sim \frac{1}{2}k^2n^2$ SWAPs and $\sim 2kn$ depth are required to implement the operations $U_{A_i,A_j}$ for $i,j \in [n], i \neq j$ on LNN.
	
\end{theorem}
\begin{proof}
    We adopt the strategy for single qubit registers from Thm.~\ref{theorem:qubit-all-to-all}. We apply permutations $\pi_1,\pi_2$ for $n$ times in total on registers $\Hl_{A_i}$.  After swapping the registers according to $\pi_1$, we apply the operations $U_{A_i,A_j}$ for the pairs $(\Hl_{A_i},\Hl_{A_{j}})$ that are swapped. Similar reasoning applies for $\pi_2$. To swap registers $\Hl_{A_i}$ and $\Hl_{A_j}$, we use the register swap from Thm.~\ref{theorem:regswap}, requiring $k^2$ SWAPs and $\sim 2k$ depth for each swap. Each $\pi_1$ and $\pi_2$ requires $\sim n/2$ register SWAPs, resulting in a total of $\sim\frac{1}{2} k^2n^2$ SWAPs and $\sim 2kn$ depth. 
\end{proof}

Given registers of the form $\Hl_{A_1}\Hl_{A_2}\cdots \Hl_{A_n}\Hl_{B_1}\Hl_{B_2}\cdots \Hl_{B_n}$, one may want to perform operations between register pairs $\Hl_{A_i}\Hl_{B_j}$. To achieve this, it is beneficial to interlace the registers first.

\begin{theorem}[Interlacing strategy]\label{theorem:reginterlace} Let
	$\Hl_{A_1}\Hl_{A_2}\cdots \Hl_{A_n}\Hl_{B_1}\Hl_{B_2}\cdots \Hl_{B_n}$ be a
	sequence of registers. Suppose $\Hl_{A_i}$ has $k$ qubits and $\Hl_{B_i}$ has
	$m$ qubits for all $i$. Then one can reorder quantum registers as
	$\Hl_{A_1}\Hl_{B_1}\Hl_{A_2}\Hl_{B_2}\cdots \Hl_{A_n}\Hl_{B_n}$ with $\sim \frac{1}{2}kmn^2$ SWAPs and $\sim \max(nk+m, nm+k )$ depth on LNN architecture.
\end{theorem}

\begin{proof}
	Let us first count the number of gates used. To move each register
	$\Hl_{B_i}$ one by one, we use the procedure proposed in
	Theorem~\ref{theorem:regswap}. To move $\Hl_{B_i}$, we pass it through
	registers $ \Hl_{A_{i+1}} \cdots \Hl_{A_{n}}$ consisting of $(n-i)k$ qubits.
	Thus the total number of SWAPs is
	\begin{equation}
	\sum_{i=1}^n (n-i)km = km\frac{(n-1)n}{2} \sim \frac{1}{2}km n^2.
	\end{equation}
	
	Let us assume $m\leq k$. We follow the strategy in Theorem~\ref{theorem:regswap}
	by first interlacing qubits from $B$ registers with qubits from $A$ registers,
	then in parallel move qubits from $B$ registers to the left, and finally
	uninterlace them. The first qubit of $\Hl_{B_1}$ has to pass $(n-1)k$ qubits,
	so we need the same number of SWAP layers. Then we need $m-1$ layers of SWAPs to
	put the other elements of $\Hl_{B_1}$ in the correct position. By the
	assumption $m\leq k$, all qubits from the remaining $B$ registers are already in
	the correct position, so the total depth is $(n-1)k+m-1 \sim nk+m$ SWAP layers.
	
	Let us now consider $m> k$ case. This is a symmetric case to the one already
	considered; however, instead of moving $B$ registers to the left, we move $A$
	registers to the right. The depth in this case will be  $\sim nm+k$, and to
	sum up the depth will be 
	\begin{equation}
	\begin{cases}
	\sim nk+m, & m\leq k, \\
	\sim nm+k, & m> k,    
	\end{cases}
	\end{equation}
	which is equivalent to $\sim \max(nk+m, nm+k)$.
\end{proof}

For the swap network required to implement operations on all register pairs $\Hl_{A_i}\Hl_{B_j}$, we will need a strategy that guarantees the appearance of each register at each position. It is enough to apply the permutation $\pi$ from Lemma~\ref{lemma:qubit-all-to-all} more times than required in there.

\begin{lemma}\label{lemma:perm} Let $\pi$ be defined as in Lemma~\ref{lemma:qubit-all-to-all}.
	Then the following is true: after applying  $n$ times permutation $\pi$ on the sequence
	$(1,\dots,n)$, each number will appear at each position at least once.
\end{lemma}
\begin{proof}
	In order to satisfy the statement of the lemma, it is enough to show that $ \pi $ is 
	a cyclic permutation (with no fixed points). Indeed, suppose that there is an element in the list that appears only on a proper subset of the available positions. But then it forms its own, smaller cycle. Thus permutation $\pi$ cannot be itself a cyclic permutation.
	
	Now let us show that $\pi$ is a cyclic permutation. Equivalently, one can write
	$\pi$ using cycle notation for odd $n$ as
	\begin{equation}
	\pi = (1,3,5,\dots,2k-1, 2k+1, 2k,2k-2,\dots,4,2),
	\end{equation}
	and for even $n$ as
	\begin{equation}
	\pi = (1,3,5,\dots,2k-1, 2k,2k-2,\dots,4,2).
	\end{equation}
	In both cases, we have a cycle consisting of all elements of $\{1,\dots,n\}$. Hence $\pi$ is a cyclic permutation with no fixed points.
\end{proof}

\begin{theorem}[Inter-all-to-all strategy]\label{theorem:inter-all-to-all} Let
	$\Hl_{A_1}\Hl_{B_1}\cdots \Hl_{A_n}\Hl_{B_n}$ be registers, with $\Hl_{A_i}$
	having $k$ qubits and $\Hl_{B_j}$ having $m$ qubits each. Let $U_{A_i,B_j}$ be commuting operations 		
	on registers $A_i,B_j$ for $i,j
	\in [n]$. Then additional $ \sim n^2 ((m+k)\min(m,k)+km)$ SWAPs and $ n(2k+2m+\min(m,k)) $ depth are required to implement the operations $U_{A_i,B_j}$ for $i,j \in [n]$ on LNN.
	
\end{theorem}
\begin{proof}
	W.l.o.g. let us assume that $k\leq m$. Then we
	make the register reordering defined by permutation $\pi$ from
	Lemma~\ref{lemma:perm}, by fixing $\Hl_{A}$ registers and moving the $\Hl_{B}$
	registers. Each permutation transposition is equivalent to the reordering
	$\Hl_{B}\Hl_A \Hl_{B'}\mapsto\Hl_{B'} \Hl_A \Hl_{B}$, which can be implemented
	as
	\begin{equation}
	\Hl_{B} \Hl_A \Hl_{B'}\mapsto \Hl_A \Hl_{B'} \Hl_{B}\mapsto \Hl_{B'} \Hl_A \Hl_{B}.
	\end{equation}
	After applying permutation each time, we apply $U_{A,B}$ on the pair of
	consecutive registers $\Hl_A \Hl_B$ (in this order).

	Let us calculate the number of gates. We apply the $\pi$ permutation $n$ times, and each application requires $\sim n$ transpositions. Finally, each transposition can be implemented through two register SWAPs, which based on Theorem~\ref{theorem:regswap} gives us $(m+k)m+km$ SWAPs. So the total number of SWAPs is $ \sim n^2 ((m+k)m+km)$.
 
	Let us now calculate the depth. We need to add the depth equivalent to two
	transpositions times $n$. Each transposition gives us the depth
	$\sim (k+m+m) + (k+m)$. So the final depth is
	$\sim n(k+m+m) + (k+m))$.

  In the symmetric case $k \leq m$, we fix $\Hl_{B}$ registers and move the $\Hl_{A}$
	registers. The total number of SWAPs in this case is $\sim n^2(m+k)m + km$ and the depth is $\sim n(2k+2m+k)$. Hence, to summarize the number of SWAPs is given by 
\begin{equation}
	\begin{cases}
	\sim n^2(m+k)k + km, & k\leq m, \\
	\sim n^2(m+k)m + km, & k> m,    
	\end{cases}
	\end{equation}
 which is equivalent to $ \sim n^2 ((m+k)\min(m,k)+km)$. Similarly, the depth is given by
 \begin{equation}
	\begin{cases}
	\sim n(2k+2m+m), & k\leq m, \\
	\sim n(2k+2m+k), & k> m,    
	\end{cases}
	\end{equation}
which can be equivalently expressed as $n(2k+2m+\min(m,k))$.
\end{proof}

\begin{corollary} The circuit presented in Theorem~\ref{theorem:inter-all-to-all} requires $\sim 3n^2k^2$ SWAPs and has depth $\sim 5nk $ provided $k\sim m$.
	
\end{corollary}

\subsection{Practical examples on implementations for LNN}\label{sec:imp-lnn}
In this section, we will investigate how to use swap network strategies to implement Prog-QAOA for Max-$K$-Cut and TSP on LNN architecture.

To implement the program for computing the objective in Max-$K$-Cut problem, we need to compare all pairs of variables $(c_i,c_j)$. Given the sequence of variables $(c_1,\dots,c_n, \mathrm{flag}_1, \dots \mathrm{flag}_n)$, we start with applying the interlacing strategy to obtain $(c_1,\mathrm{flag}_1,\dots,c_n,\mathrm{flag}_n)$. Then we apply intra-all-to-all strategy so that the registers corresponding to the pair of variables $(c_i,\mathrm{flag}_i)$ and $(c_j,\mathrm{flag}_j)$ are neighboring at some point. We compare pair-wise neighbouring variables $(c_i,c_j)$, by first storing on $\mathrm{flag}_i$ the value of $c_i \neq c_j$, and then yielding value of $\mathrm{flag}_i$, and uncomputing value of $\mathrm{flag}_i$.  To revert back the original order of the variables, optionally, we undo the all-to-all network. Finally, we un-interlace the variables to get back the variables in the original order $(c_1,\dots,c_n, \mathrm{flag}_1, \dots \mathrm{flag}_n)$ or $(c_n,\dots,c_1, \mathrm{flag}_n, \dots \mathrm{flag}_1)$, depending whether we performed intra-all-to-all uncomputation or not. Note that interlacing strategy can be omitted if we already start with interlaced variables, depending also on the implementation of the mixer.

For TSP, to implement the algorithm for permutation checking, we start with variables in the order $(c_0,\dots, c_{n-1}, \mathrm{count}_0, \dots, \mathrm{count}_{n-1}, \mathrm{flag})$. Next, we interlace the variables $c_i$ and $\mathrm{count}_i$. Then we apply inter-all-to-all strategy and implement $\mathrm{count}_i \gets \mathrm{count}_i +1 \mod 2$, conditioned the neighboring $c_j$ is equal to $i$ . After this step, we uninterlace $count_i$ and $c_i$. Then we implement AND operation (multi-controlled NOT) with controls on $\mathrm{count}_i$ and target on flag. Finally, we return the value of flag and uncompute the whole part before returning flag.

To implement the objective value of TSP, we start with the ordering 
\begin{equation}
    (c_0,\dots, c_{n-1}, \mathrm{flag}_0, \mathrm{edge}_0,\dots, \mathrm{flag}_{n-1}, \mathrm{edge}_{n-1}).
\end{equation} 
Next, we interlace $c_i$ with the pair of variables $(\mathrm{flag}_j,\mathrm{edge}_j)$. Then we apply the inter-all-to-all strategy to set the value of edge$_i$ using $\mathrm{flag}_i$ as the flag. Note that in this step the value of edge$_i$ depends both on surrounding $c_j$ and $c_{j+1}$, however since both of them remain fix during the inter-all-to-all network, the operation of setting edge$_j$ remains local.  Afterwards, we un-interlace $c_i$ and $(\mathrm{flag}_j,\mathrm{edge}_j)$. For each pair, we first set the value of $\mathrm{flag}_j$ to $\mathrm{edge}_j = i$ for each $i\neq j$, and then we yield the value of $\mathrm{flag}_j$. This process can be implemented in parallel for all $j$. Finally, we uncompute all variables.

Note that the qubits used for storing count$_i$ can be reused in the objective value calculation for the auxiliary variables. Note that while naturally, we think in terms of the logical variables, at the level of the physical device, attachment of physical qubits to logical variables can be different from program to program. Furthermore, after each step, the auxiliary variables end at $\ket{0}$ (or at worst in a predetermined) state. It is important to note that the qubits corresponding to the main variables which are optimized should be in such a position that can be reused for the next program or mixer. The original placement of the physical qubits would be valid, although potentially it may be inefficient. Finally, clearly some improvements can lead to a severe decrease in the resources taken, especially in terms of the number of gates for example by avoiding unnecessary interlacing or using the relative-phase implementation of the permutation gates like iSWAPs. However, this type of improvement will not lead to a further decrease in complexity, therefore they were not considered here.

\subsection{Prog-QAOA for TSP on LNN} 

In this section, we describe the quantum circuit suitable for the LNN architecture that implements Algorithms~\ref{alg:permutation} and~\ref{alg:tsp-costroute}. Note that for all-to-all architecture, we can consider exactly the same implementation, but SWAPs are applied on the logical qubits. Derivation of the cost metrics and implementation details can be found in Appendix Sec.~\ref{sec:cost-tsp}.

We start in an equal superposition of all the valid cities for each register $\ket{b_t}$
\begin{equation}
\bigotimes_{t=0}^{n-1} \frac{1}{\sqrt n}\sum_{i=0}^{n-1}\ket i.
\end{equation}
This state can be efficiently prepared by adapting the procedure presented in~\cite{DaftWullie2022stack}. With such initial state preparation, we use the Grover mixer implementation presented in~\cite{bartschi2020grover}. 

Algorithm~\ref{alg:permutation} can be implemented as follows. First, we add $n$ 1-qubit ancilla to the end. Then, we interlace those $n$ qubits with $n$ registers $\ket{b_t}$ using the interlacing strategy. Next, we use the all-to-all strategy, and for each register $\ket{b_t}$ and ancilla qubit $\ket{c_j}$, we apply NOT on $\ket{c_j}$ if $\ket{b_t}$ is equal to $j$. This can be done by computing bit-wise XOR $\ket{b_t\oplus j}$, and then using the multi-controlled NOT gate, which applies NOT on $\ket{c_j}$ if all qubits in  $\ket{b_t\oplus j}$ are set to $\ket{0}$. After this step, bit-wise XOR is uncomputed, and the all-to-all
strategy continues. After the all-to-all strategy is completed, we undo the interlacing strategy so that all $\ket{c_j}$ are neighboring again. Note that at this moment, $\ket{c_j}$ is set to $\ket1$ if there is an odd number of $\ket{b_t}$ equal to $\ket{j}$. Therefore it is enough to check whether all $\ket{c_j}$ are equal to $\ket{1}$. This can be done by a multi-controlled NOT gate applying NOT on the extra target qubit $\ket{\rm flag}= \ket{0}$ if all $\ket{c_j}$ are set to $\ket1$. After this, we apply NOT and then $Z$-rotation on $\ket{\rm flag}$ (NOT is
required to add local phase to $\ket{0}$ which indicates incorrect permutation), and the whole computation except for the $Z$-rotation is uncomputed.

Algorithm~\ref{alg:tsp-costroute} can be implemented as follows. First we add $n$ 1-qubit registers $\ket{f_i}$ and $n$ $\lceil\log n \rceil$-qubit registers $\ket{{\rm edge}_i}$. We interlace $\ket{b_t}$ with newly created registers, so that we have $n$ consecutive registers in order $\ket{b_i,f_i,{\rm edge}_i}$. Then we apply the all-to-all strategy by swapping the $\ket{{\rm edge}_i}$ registers. Assuming the neighbouring registers are $\ket{b_t,f_t,{\rm edge}_i,b_{t+1},f_{t+1},{\rm edge}_j}$ we do the following. First we store on $\ket{f_t}$ the information whether $\ket{b_t}=\ket{i}$. This can be done using the same quantum circuit as for permutation verification. Then provided $\ket{f_t}=\ket{1}$ we make a copy from $\ket{b_{t+1}}$ to $\ket{{\rm edge}_i}$. Note that if $\{b_t\}$ forms a permutation, then such copy is applied once and $\ket{{\rm edge}_i}$. This can be done with
multi-controlled NOT gate controlled on $\ket{f_t}$ and qubits from $\ket{b_{t+1}}$, and target on qubits from $\ket{{\rm edge}_i}$. Then we uncompute the $\ket{f_t}$ estimation. Provided $\ket{b_t}$ forms a permutation, in each $\ket{{\rm edge}_i}$, we store the information about which city is visited next from city $i$. After the all-to-all strategy is completed, for each register $\ket{{\rm edge}_i}$ and all
cities $j\neq i$ we store the information whether $\ket{{\rm edge}_i}=\ket{j}$ on neighbouring flag qubit $\ket{f_t}$, and we apply $Z$ rotation proportional to $W_{i,j}$ on $\ket{f_t}$. Then we uncompute $\ket{f_t}$. After all $W_{i,j}$ are included, we undo the all-to-all strategy and interlacing.

Note that with the strategy above, we cannot include the information between $\ket{b_{t+1}}$ and $\ket{b_1}$. This can be easily solved by creating an additional register $\ket{b_1'}$, moving it next to $\ket{b_1}$, copying the value from $\ket{b_1}$ to $\ket{b_1'}$ with controlled NOTs and then moving it after the $\ket{b_{n-1},f_{n-1},{\rm edge}_{n-1}}$. This way, when $\ket{{\rm edge}_i}$ during the all-to-all strategy will be placed at the end, it will have access to necessary information from both $\ket{b_{n-1}}$ and $\ket{b_1}$.

If $\{b_t\}$ does not form a permutation or $\ket{b_t}$ does not represent a valid city (like in the case of $X$-mixer), a copy from $\ket{b_{t+1}}$ to
$\ket{{\rm edge}_i}$ may take place more or less than once. In such a case, $\ket{{\rm edge}_i}$ may point at an incorrect or even a non-existing city. Taking into account how the edge cost is included, it can either result in not adding the cost at all or adding an incorrect one. Therefore, for $\{b_t\}$ not being a permutation, the energy from the cost route may range from $0$ (if for example, all $b_t$ represents an invalid city) to $n\max_{i\neq j}W_{i,j}$ (if the route moves back and forth along the most costly edge). This effect can be suppressed by an
appropriate choice of penalty value $A$ for permutation verification.

\section{Cost analysis for Max-$K$-Cut and TSP}

\subsection{Max-$K$-Cut}\label{sec:cost-maxcut}
For the analysis, we assume the graph $G=(V,E)$ is complete and weighted with $n=\#V$ nodes. The same results apply when the number of edges is $\Omega(n^2)$. We assume $K< n$, which is necessary to construct nontrivial problem instances. 
We assume $W:E\to \RR_{\geq 0}$ is the function of weights. We assume $\|W\|_\infty =\Theta(1)$. For simplicity, we will write $W_{i,j} \coloneqq W(\{i,j\})$.

Some of the Hamiltonians require choosing a penalty parameter $A$. We assume that $A=\Theta(1)$.

\subsubsection{X-QAOA}\label{sec:app-maxcut-qaoa}
The QUBO takes the form~\cite{tabi2020quantum}
\begin{equation}
A \sum_{i=1}^n \bigl( 1 - \sum_{k=1}^K b_{i,k} \bigr)^2 + \sum_{\substack{i,j=1\\i<j}}^n  \sum_{k=1}^K W_{i,j} b_{i,k} b_{j,k}.
\end{equation}
We start in a superposition of all states, and the mixer is $X$-rotations applied on all qubits.

The number of qubits (spins) and the effective space size are $nK$. The number of 2-local Pauli terms, and thus gates and parameterized gates is at most $n\binom{K}{2}+  \binom{n}{2} K = \order{n^2K}$. Note that 1-local Pauli terms are insignificant for this analysis. For depth implementation note that the first part can be implemented with depth $\order{K}$, because for each $i$ one can implement $\bigl( 1 - \sum_{k=1}^K b_{i,k} \bigr)^2$ independently, and each such term can be implemented in $\order{K}$ \cite{o2019generalized}. Similarly one can implement the objective part, as for each fixed $k$, the part $\sum_{\substack{i,j=1\\i\neq j}}^n W_{i,j} b_{i,k} b_{j,k}$ can be implemented in $\order{n}$ depth. Thus the total depth is $\order{n+K} = \order{n}$. Note that the depth estimation is tight given the number of qubits and gates.

On LNN, one can implement the cost Hamiltonian above using the swap network \cite{o2019generalized} which results in $\order{nK}$ depth. To show this is optimal, we will use the Theorem~\ref{theorem:depth-lowerbound-lnn-constant-radius}. It is enough to show that any pair of spins can be connected with a path of length at most 2 on the graph representing the QUBO (2-local cost Hamiltonian) interactions. Let $(s_{i,k},s_{i',k'})$ be pair of different spins. If $i=i'$ or $k=k'$, then spins are connected by an edge, respectively thanks to the constraint part and the objective part. Otherwise, the path is $s_{i,k} \to s_{i,k'} \to s_{i',k'}$.

For the energy span, note that we have 
\begin{equation}
\begin{split}
\max_b \left(A \sum_{i=1}^n \bigl( 1 - \sum_{k=1}^K b_{i,k} \bigr)^2 + \sum_{\substack{i,j=1\\i< j}}^n W_{i,j} \sum_{k=1}^K b_{i,k} b_{j,k}\right) 
&\leq A n (K-1)^2 + \sum_{\substack{i,j=1\\i< j}}^n W_{i,j} K \\
&\leq AnK^2 + \binom{n}{2}K \|W\|_\infty,
\end{split}
\end{equation}
which ultimately gives us the upper bound $\order{n^2K}$ for the maximum energy. This is tight which can be shown by taking $b_{i,k}\equiv 1$. Note that the minimal energy is positive and at most $n\|W\|_\infty=\order{n}$ which can be shown by fixing all colors to 1. This shows the energy span is $\order{n^2 K-n} =\order{n^2K}$, which is tight.

The initial state preparation and mixer consists of single-qubit gates only; thus, their costs are negligible.

\subsubsection{XY-QAOA}\label{sec:xy-qaoa-maxkcut}
The cost Hamiltonian takes the form 
\begin{equation}
\sum_{\substack{i,j=1\\i<j}}^n W_{i,j} \sum_{k=1}^K b_{i,k} b_{j,k}.
\end{equation}
We start in a product of $n$, $K$-qubit $W$-state, and the mixer is the partitioned ring or complete Hamiltonian consisting of $X_iX_j + Y_iY_j$ applied on each register $\{b_{i,1},\dots,b_{i,K}\}$~\cite{wang2020x}.

The number of qubits (spins) is $nK$, and the effective space size is $\log(K^n) = n\log K$. The number of Pauli terms, and thus gates and parameterized gates is at most $ \binom{n}{2} K = \order{n^2K}$. The depth equals $\order{n}$ taking the approach presented in Sec.~\ref{sec:app-maxcut-qaoa}. Note that the estimation is tight given the number of qubits and gates.

Let us consider the depth on the LNN architecture. Let as assume that $XY$-mixers and $ZZ$ interactions can be implemented in any order for a fixed QAOA level. Taking into account the number of gates and qubits, the lower bound on the minimal depth is $\order{n^2K/(nK)} = \order{n}$. One can use the following strategy, which achieves only $\order{nK}$ depth. We start in order $b_{1,1},b_{2,1},\dots,b_{n,1},b_{1,2},\dots$, implement the objective Hamiltonian using \cite{o2019generalized} on $\{b_{1,k},\dots,b_{n,k}\}$ for each $k=1,\dots, K$. Then we re-order the qubits, so that we result in order $b_{1,1},b_{1,2},\dots,b_{1,n},b_{2,1},\dots$, and apply $XY$-mixer in depth $\order{1}$. The reordering requires $\order{nK}$ depth~\cite{saadeh2019performance} which dominates the total depth.

One can show that
\begin{equation}
\sum_{\substack{i,j=1\\i<j}}^n  \sum_{k=1}^K W_{i,j} b_{i,k} b_{j,k} \leq \sum_{\substack{i,j=1\\i<j}}^n W_{i,j}  =\order{n^2},
\end{equation}
which is achieved by choosing the same color for all nodes.

The initial state preparation can be implemented in $\order{K}$ depth even on the LNN network \cite{bartschi2019deterministic}, which is negligible compared to other parts of the circuit. The mixer cost is negligible as shown above.

\subsubsection{HOBO} \label{sec:maxcut-hobo}
The Hamiltonian takes the form~\cite{tabi2020quantum}
\begin{equation}
\begin{split}
H(b) &= A\sum_{i=1}^{n} H_{\rm valid} (b_i) + \sum_{\substack{i,j=1\\i> j}}^{n} W_{i,j} H_{\delta}(b_i,b_j),
\end{split}
\end{equation}
where $b_i$ are registers consisting of $b_{i,1},\dots,b_{i,\lceil \log K\rceil }$, encoding colors $\{0,1\dots,K-1\}$ in binary. The first term introduces a penalty if $b_i$ encodes some color $\geq K$ and $H_{\rm valid}$ and $H_\delta$ are implemented as in \cite{glos2020space}. The initial state and mixer are the same as in Sec.~\ref{sec:app-maxcut-qaoa}.

The number of qubits and the size of the effective space are $\order{n \log K}$. $H_{\rm valid}$ consists of at most $2^{\lceil \log K\rceil}=\order{K}$ Pauli terms, as the Hamiltonian is defined over $\lceil \log K\rceil$ qubits. Each such Pauli operator can be implemented using $\sim 2\log K$ gates using the decomposition from \cite{seeley2012bravyi}. For the objective Hamiltonian, using the Gray code approach from \cite{glos2020space}, each Pauli term can be applied using 2 CNOT gates and one parameterized gate, and we have $2^{\lceil \log K\rceil}$ Pauli terms~\cite{glos2020space,campbell2021qaoa}. So the total number of gates is at most $\order{nK\log K + \binom{n}{2} K} = \order{n^2 K}$. The number of parameterized gates is the same.

Let us consider the depth cost. $H_{\rm valid}$ can be implemented in parallel for each $b_i$ register, and can be implemented in depth $\order{K \log K}$, assuming we apply each gate one by one. For LNN architecture, since all qubits are at LNN-distance $\lceil \log K\rceil$, we can implement $H_{\rm valid}$ in $\order{K \log^2 K}$ depth. For the objective part for all-to-all architecture, we apply the swap-network order~\cite{o2019generalized} on the registers, which results in $\order{n}$ depth times the depth of $H_\delta(b_i,b_j)$. $H_\delta$ can be implemented in $\order{K}$ depth~\cite{glos2020space} which results in $\order{nK}$ depth overall. For LNN architecture, we need extra $\order{\log K}$ for swapping a pair of registers in the swap network and $\order{\log K}$ to make interacting qubits neighboring while implementing each CNOT in the $H_\delta$. This gives depth $\order{n(\log K +K \log K)} = \order{nK\log K}$ for the objective Hamiltonian, and the same for the whole Hamiltonian. Note that the estimations for both architectures are almost tight, given the number of qubits and gates.

For the energy span, we can see that
\begin{equation}
\begin{split}
H(b) &= A\sum_{i=1}^{n} H_{\rm valid} (b_i) + \sum_{\substack{i,j=1\\i> j}}^{n} W_{i,j} H_{\delta}(b_i,b_j) \leq An \lceil \log K\rceil + \binom{n}{2} \|W\|_\infty = \order{n^2},
\end{split}
\end{equation}
where the estimations for $H_{\rm valid}$ and $H_\delta$ come from  \cite{glos2020space}. This can be achieved by choosing the same color for all nodes. 

The initial state preparation and mixer consist of single-qubit gates only; thus, their costs are negligible.

\subsubsection{Fuchs-QAOA} \label{sec:maxcut-orig-sim-qaoa}

\begin{figure}\centering 
	\includegraphics[scale=.9]{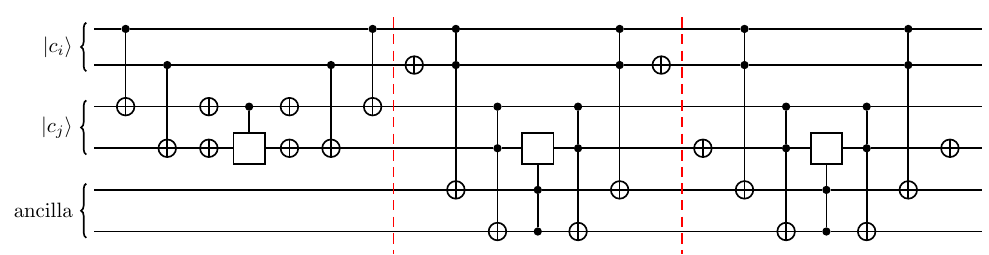}
	\caption{\label{fig:fuchs} Implementation of Fuchs-QAOA on registers $\ket{c_i}$ and $\ket{c_j}$ for $K=3$ colors. Note that the first part of the circuit penalizes the state if both colors are the same. The second step is the fixing part for the color combination $\ket{2,3}$, and the third step is the fixing part for the color combination $\ket{3,2}$. The last two parts are required as both $2$ and $3$ denote the third color according to Fuchs-QAOA~\cite{fuchs2021efficient}. }
\end{figure}

In the Fuchs-QAOA~\cite{fuchs2021efficient}, the Hamiltonian is implemented as in Fig.~\ref{fig:fuchs}. We have $n$ registers, each having $\lceil \log K\rceil$ qubits, with the same interpretation as in Sec.~\ref{sec:maxcut-hobo} for $K$ being a power of 2. If $K$ is not a power of 2, then the remaining colors also denote the $K$-th color. The procedure goes as follows: first, the objective function is implemented assuming all $2^{\lceil \log K \rceil}$ colors are different, which we call \emph{the objective phase}. Then the undeserved cases, where integers denote the same $K$-th color, are fixed, which we call \emph{the fixing phase}. This correction is particularly harmful when $K=2^l+1$ for some integer $l$, which is the worst-case scenario. The initial state and mixer are the same as in Sec.~\ref{sec:app-maxcut-qaoa} and Sec.~\ref{sec:maxcut-hobo}.

The number of qubits required and the effective space size is $\order{n\log K}$. For each edge, the objective phase uses $\order{\log K}$ CNOTs and 1-qubit gates and a single $\sim \log K$-controlled $Z$ rotation. The controlled $Z$-rotation can be implemented with the decomposition presented in Lemma 7.11 in \cite{barenco1995elementary}, where each Toffoli is implemented using the ancilla implementation~\cite{he2017decompositions} using $\order{\log K}$ gates. Note that this implementation uses only $\sim \log K$ ancilla qubits; hence does not change the previous estimations. So for the whole objective phase, we need $\order{\binom{n}{2}\log K} = \order{n^2 \log K}$ gates and $\order{n^2}$ parameterized gates. For each edge, the fixing phase requires $2\binom{2^{\lceil \log K\rceil} - (K-1)}{2} $ fixing circuits, which in the worst-case scenario grows like $\Theta(\binom{K}{2}) = \Theta(K^2)$. For every single fixing circuit, we need 2-controlled rotation, $\order{\log K}$ 1-qubit gates, and 4 $\log K$-controlled NOTs, for which we can use the same implementation as before. This, in total, will give $\order{n^2K^2\log K}$ gates for the whole fixing phase in the worst-case scenario, which dominates the cost. Let us note that when $K$ is a power of two, the fixing circuits are not implemented, and the whole implementation is much more efficient.

Finally, using the swap network order from \cite{o2019generalized}, one can show the objective phase can be implemented in $\order{n \log \log K }$ depth, and the fixing phase can be implemented in depth $\order{nK^2\log\log K}$, as fixing circuits need to be implemented one by one. This is tight given the number of qubits and gates.

On LNN architecture, we use the swap network from \cite{o2019generalized} to swap registers for different vertices, which costs $\order{n\log K}$ taking into account only the cost of swaps. For each edge, the objective phase proceeds as follows: first, we interlace qubits as in Theorem~\ref{theorem:reginterlace} in depth $\order{\log K}$ and apply CNOT gates. Then we uninterlace qubits with the same depth. The rest of the objective phase is implemented as previously, but we choose the no-ancilla implementation from~\cite{saeedi2013linear,da2022linear} requiring $\order{\log K}$ depth. So, each edge of the objective phase can be implemented in $\order{\log K}$ depth. Each fixing circuit is implemented by interlacing two ancilla qubits with the registers, applying Toffoli gate as in~\cite{saeedi2013linear,da2022linear}, and uninterlace, which produces $\order{\log K}$ depth. Then 2-controlled rotation can be implemented in $\order{1}$. The whole circuit has depth $\order{nK^2\log K}$.

For the energy span, observe that all samples from the Hamiltonian are valid solutions, but only some of them appear more often. Thus the energy difference is the same as in Sec.~\ref{sec:xy-qaoa-maxkcut}, which is $\order{n^2}$.

The initial state preparation and mixer consists of single-qubits gates only; thus, their cost is negligible.

\subsubsection{Prog-QAOA}\label{sec:sim-maxkcut} 

The Prog-QAOA is described in the main article in Methods. The implementation essentially leads to the same circuit as for the Fuchs-QAOA from Sec.~\ref{sec:maxcut-orig-sim-qaoa}, except for the initial state preparation and the mixer, plus there is no need to implement \emph{the fixing phase}. Thus the quality measures can be derived from Sec.~\ref{sec:maxcut-orig-sim-qaoa}. 

The number of required qubits is $\order{n \log K}$. The effective space size is $\log(K^n) = n\log K$. The number of gates is $\order{n^2 \log K}$ The number of parameterized gates is $\order{n^2}$. The depth is $\order{n \log \log K}$ for all-to-all architecture and $\order{n\log K}$ for LNN connectivity. The energy span is $\order{n^2}$.

The initial state is the product of $n$ quantum states $\frac{1}{\sqrt K} \sum_{k=0}^{K-1} \ket k$. Each such superposition will be implemented independently; so let us focus on generating only a single such register. The mentioned state can be prepared by adjusting the procedure taken from \cite{DaftWullie2022stack}, which is visualized for the case $n=14$ in Fig.~\ref{fig:maxkcut-inital-state}. Let $m\coloneqq \lceil \log K\rceil$. The procedure goes as follows: First, we change the qubit $\ket{0}$ into $\sqrt{2^{m-1}/K}\ket 0 +\sqrt{(K-2^{m-1})/K}\ket 1$ through appropriate $Y$ rotation. Then we create a uniform superposition of the states of the form $\ket{0?\cdots?}$ by applying a controlled Hadamard if the state of the first qubit is $\ket{0}$. Hence, we only need to correct the amplitude for basic states starting with $\ket{1}$. We do this by a similar controlled rotation on the next qubit, corresponding to the next digit, which can be $0$ or $1$ in the binary representation of numbers below $K$, with the already fixed prefix. For example, in Fig.~\ref{fig:maxkcut-inital-state} in the second block, we rotate qubit $main_1$, however for $11=1011_2$ it would be the third qubit, as there is no number smaller than $11$ which has the form $11??_2$. This procedure is repeated at most $\lceil\log K\rceil$ times, as in each case, we rotate a different qubit. The procedure stops when the number of left basic states is of the form $2^l$. If it is 1, then nothing is done. Otherwise, we apply controlled Hadamards with the target on the last $l$ qubits and controls on all so far rotated qubits.

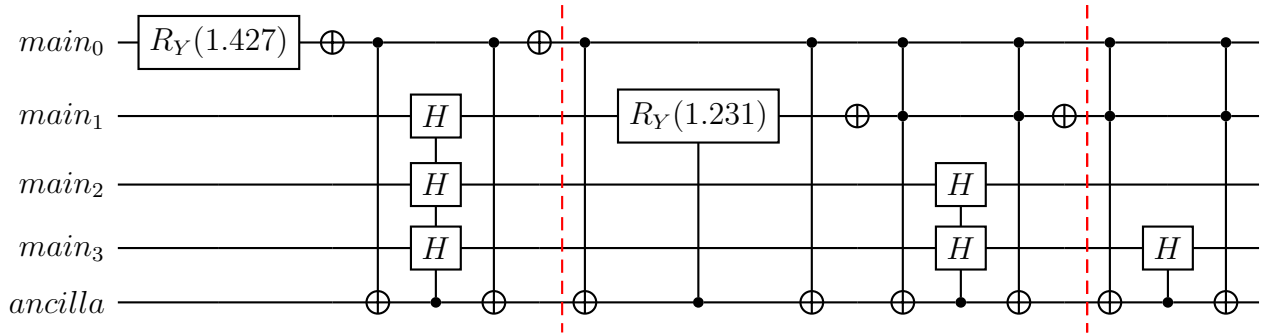
\begin{figure}\centering
	
	\begin{quantikz}[row sep=1.5ex, column sep=1.5ex]
		\lstick{$main_0$} & \gate{R_Y(1.427)} & \targ{} & \ctrl{4} & \qw & \ctrl{4} & \targ{} \slice{} & \ctrl{4} & \qw & \ctrl{4} & \qw & \ctrl{1} & \qw & \ctrl{1} & \qw \slice{} \slice{}  & \ctrl{1} & \qw & \ctrl{1}& \qw \\
		\lstick{$main_1$} & \qw & \qw & \qw & \gate H & \qw & \qw & \qw & \gate{R_Y(1.231)} & \qw & \targ{} & \ctrl{3} & \qw & \ctrl{3} & \targ{} & \ctrl{3} & \qw & \ctrl{3} & \qw \\
		\lstick{$main_2$} & \qw & \qw & \qw & \gate H & \qw & \qw & \qw & \qw & \qw & \qw & \qw  & \gate H & \qw & \qw & \qw & \qw & \qw & \qw \\
		\lstick{$main_3$} & \qw & \qw & \qw & \gate H & \qw & \qw & \qw & \qw & \qw & \qw & \qw  & \gate H & \qw & \qw & \qw & \gate H & \qw & \qw \\
		\lstick{$ancilla$} & \qw & \qw & \targ{} & \ctrl{-3} & \targ{} & \qw & \targ{} & \ctrl{-3} & \targ{} & \qw & \targ{} & \ctrl{-2} & \targ{} & \qw & \targ{} & \ctrl{-1} & \targ{} & \qw 
	\end{quantikz}
	\caption{\label{fig:maxkcut-inital-state} Quantum circuit preparing the state $\frac{1}{\sqrt K}
		\sum_{k=0}^{K-1} \ket k$ for $ K=14 $.} 
\end{figure}

The whole procedure consists of at most $\lceil \log K\rceil$ steps, each having a constant number of (perhaps multi-controlled) gates. Multi-controlled NOTs can be implemented in depth $\order{\log K}$ and $\order{\log^2 K}$ gates even on LNN architecture~\cite{saeedi2013linear,da2022linear}. The additional cost for LNN architecture arises from moving the ancilla qubits, which produces $\order{\log K}$ cost for each iteration. Furthermore, reordering the control qubits so that they are neighboring may also require $\order{\log K}$ depth (we assume that the bottom-most controlled qubit remains at its place). Thus the total depth is $\order{\log^2 K}$ for both all-to-all and LNN architecture, and the number of gates is $\order{n \log^3 K}$. This cost is negligible compared to the objective implementation.

For the mixer, we use the Grover mixer Hamiltonian presented in~\cite{bartschi2020grover}. Let $U_S$ be the unitary that prepares the initial state as before. The mixer consists of implementing $U_S^\dagger$, the layer of NOT gates on all qubits, multi-controlled $Z$-rotation, another layer of NOT gates, and $U_S$. Based on Theorem 2 from~\cite{bartschi2020grover}, the number of gates and depth, even on LNN, are equal to the cost of implementing $U_S$ plus $\order{n\log K}$, which is $\order{n\log K}$ for depth and $\order{n\log^3 K}$ for the number of gates in total. The number of parameterized gates is $\order{n\log n}$~\cite{bartschi2020grover}.

\subsection{TSP}\label{sec:cost-tsp} We will consider the TSP problem defined over $n$ cities. Note
that in most cases, we can reduce the problem into a problem over $n-1$ cities by fixing that the first city is visited at time $1$. This simplification does not have a significant impact on the estimation; thus, it is omitted.

$W$ is a function describing the cost between cities. For simplicity we will write $W_{i,j} \coloneqq W(\{i,j\})$. We assume $\|W\|_\infty = \Theta(1)$. Some of the Hamiltonians require choosing a penalty parameter $A$. Some of the Hamiltonians require choosing penalty parameters $A_i$. We assume that
$A_i=\Theta(1)$.

\subsubsection{X-QAOA} The in-depth analysis of X-QAOA can be found
in~\cite{glos2020space}; here, we only recall these results. The Hamiltonian takes the form
\begin{equation}
A_1 \sum_{t={0}}^{{n-1}} \left( 1 -\sum_{i={0}}^{{n-1}} b_{t,i} \right)^{\!\!2} + A_2 \sum_{i={0}}^{{n-1}} \left( 1 -\sum_{t={0}}^{{n-1}} b_{t,i} \right)^{\!\!2} \\
+ \sum_{\substack{i,j={0}\\i\neq j}}^{{n-1}} W_{i,j} \sum_{t={0}}^{{n-1}} b_{t,i}b_{t+1,j}.
\end{equation}
The initial state is a superposition of all states, and the mixer is $X$-rotation applied on all qubits. The number of qubits and the effective space size is $\order{n^2}$. The number of 2-local $ZZ$ terms were estimated to be $\order{n^3}$~\cite{glos2020space}, and since each term requires 3 gates and one parameterized gate, the number of (parameterized) gates is $\order{n^3}$. The depth on all-to-all architecture is $\order{n}$, which is optimal given the number of gates and qubits. There exists a circuit that implements the Hamiltonian in depth $\order{n^2}$ on LNN, based on the results from \cite{o2019generalized}. At the same time, $\order{n^2}$ depth is required in light of Theorem~\ref{theorem:depth-lowerbound-lnn-constant-radius} because the QUBO graph has radius 2, thanks to the constraint Hamiltonians. The maximal energy attained is $\order{n^3}$, which is tight.

The initial state preparation and mixer consist of single-qubit gates only; thus, their costs are negligible.

\subsubsection{XY-QAOA}
The Hamiltonian takes the form
\begin{equation}
A \sum_{i={0}}^{{n-1}} \left( 1 -\sum_{t={0}}^{{n-1}} b_{t,i} \right)^{\!\!2} \\
+ \sum_{\substack{i,j={0}\\i\neq j}}^{{n-1}} W_{i,j} \sum_{t={0}}^{{n-1}} b_{t,i}b_{t+1,j}.
\end{equation}
The initial state and mixer are as in Sec.~\ref{sec:xy-qaoa-maxkcut}. 

The number of qubits, gates, and parameterized gates is the same as for X-QAOA, following very similar reasoning as in~\cite{glos2020space}. The effective space is $\log(n^n)=n\log n$. Same as for the X-QAOA, the depth on all-to-all connectivity is $O(n)$ following the same procedure as in \cite{glos2020space}. The depth on LNN is at most $\order{n^2}$ thanks to the swap network~\cite{o2019generalized}. By using Theorem~\ref{theorem:depth-lowerbound-lnn-constant-radius} we can show $\order{n^2}$ is tight: given two bits $b_{t,i}$, $b_{t',i'}$, one can find a finite path $b_{t,i} \to b_{t',i} \to b_{t'+1,j} \to b_{t',i'}$, where $j\not \in \{i,i'\}$, where the first arc comes from the constraint, and the last two arcs come from the objective part.

The maximal energy is at most $\order{n^2}$. To show this, let us first assume that there are $n_0$ indices $i$ for which $\sum_{t={0}}^{{n-1}} b_{t,i} -1 < 0 $, which is equivalent to saying $\sum_{t={0}}^{{n-1}} b_{t,i} -1  = -1$. W.l.o.g., let us assume that these indices are $i=0,\dots,n_0-1$.  Then we have 
\begin{equation}
\begin{split}
A \sum_{i={0}}^{{n-1}} \left( \sum_{t={0}}^{{n-1}} b_{t,i} -1 \right)^{\!\!2} 
&= A \sum_{i={0}}^{n_0-1} \left( \sum_{t={0}}^{{n-1}} b_{t,i} -1 \right)^{\!\!2} +A \sum_{i=n_0}^{n-1} \left( \sum_{t={0}}^{{n-1}} b_{t,i} -1 \right)^{\!\!2}\\
&\leq A \sum_{i={0}}^{n_0-1} 1 +A \left( \sum_{i=n_0}^{n-1} \sum_{t={0}}^{{n-1}} b_{t,i} -(n-n_0) \right)^{\!\!2}\\
&\leq A n_0 +A (n-n_0)n \leq 2An^2,\\
\end{split} 
\end{equation}
where in the second line we used the fact that the sum of squares of nonnegative numbers is not larger than the square of their sum.
The objective part can be upper-bounded as 
\begin{equation}
\begin{split}
\sum_{\substack{i,j={0}\\i\neq j}}^{{n-1}} W_{i,j} \sum_{t={0}}^{{n-1}}  b_{t,i}b_{t+1,j} &\leq  \|W\|_\infty \sum_{t={0}}^{{n-1}}  \sum_{\substack{i,j={0}\\i\neq j}}^{{n-1}}    b_{t,i}b_{t+1,j} =  \|W\|_\infty \sum_{t=0}^{n-1} 1 = \|W\|_\infty n.
\end{split} 
\end{equation}
This shows that the maximal energy is at most $\order{n^2}$. This can be achieved by choosing $b_{t,0}\equiv1$ and 0 for other binary variables.

Following the reasoning from Sec.~\ref{sec:xy-qaoa-maxkcut}, the costs of the initial state and the mixer are negligible.

\subsubsection{GM-QAOA}
The Hamiltonian takes the form
\begin{equation}
\sum_{\substack{i,j={0}\\i\neq j}}^{{n-1}} W_{i,j} \sum_{t={0}}^{{n-1}} b_{t,i}b_{t+1,j},
\end{equation}
and the mixer is assumed to be the Grover-mixer, with the initial state being the superposition of permutations as in~\cite{bartschi2020grover}.

The number of qubits, gates, and parameterized gates is the same as for X-QAOA, following very similar reasoning. The effective space is $\log(n!)=\Theta(n\log n)$ thanks to the Stirling formula. While the objective Hamiltonian can be implemented in depth $\order{n}$, the depth $\order{n^2}$ is required for the mixer, which is true also for LNN~\cite{bartschi2020grover}. We were not able to prove the optimality of the result for depth. However, for all-to-all architecture, one would need to find a better implementation of finding the superposition of permutations, while for the LNN network, one would additionally need to find a better implementation of $\sim n^2$-controlled $R_Z$ operation. The latter seems particularly unlikely, as one would need to distill information from qubits located at distance $\sim n^2$. 

The maximal energy attained is $\order{n}$ as the evolution is taking place only within the space of all routes, thus the maximal possible energy is at most $n\|W\|_\infty = \order{n}$.

\subsubsection{HOBO}
The Hamiltonian takes a general form
\begin{equation}\label{eq:general-tsp-hobo}
\begin{split}
H(b) &= A_1\sum_{t=0}^{n-1} H_{\rm valid} (b_t) + A_2\sum_{t=0}^{n-1} \sum_{t'=t+1}^{n-1} H_{\neq}(b_t,b_{t'}) + \sum_{\substack{i,j=0\\i\neq j}}^{n-1} W_{i,j} \sum_{t=0}^{n-1}
H_\delta (b_t,i)H_\delta (b_{t+1},j).
\end{split}
\end{equation}
Details of the implementations and detailed analysis of Hamiltonians $H_{\rm valid}$, $H_{\neq}$ and $H_\delta$ can be found in \cite{glos2020space}. The initial state is the superposition of all states, and the mixer is $X$-rotations applied on all qubits. 

The number of qubits and the effective space size is $\order{n \log n}$. Despite having higher-order terms, the number of gates and parameterized gates is equal to the number of Pauli terms which is $O(n^3)$ following the Gray code procedure~\cite{glos2020space}. We would like to point out that both $H_{\neq}$ and $H_\delta$ are greatly simplified when transformed into the cost Hamiltonian, which results in dropping a number of polynomial terms $\order{n^4}$ in the HOBO above into the number of Pauli terms $\order{n^3}$ as shown in~\cite{glos2020space,campbell2021qaoa}. The depth is $\order{n^2}$ for all-to-all architecture~\cite{glos2020space}, and using the generalized swap network \cite{o2019generalized} on registers $\{b_{t,i}:i\in [n]\}$, one can implement this with the depth $\order{n^2 \log^2 n}$ on LNN, applying similar fixes as in Sec.~\ref{sec:maxcut-hobo}. The maximal energy is $\order{n^2}$, which is tight~\cite{glos2020space}.

The initial state preparation and mixer consist of single-qubit gates only; thus, their costs are negligible.

\subsubsection{ILP}

We use the Miller-Tucker-Zemlin formulation~\cite{miller1960integer} defined as
\begin{align}
\min &\sum_{i=0}^{n-1} \sum_{\substack{j=0\\j\ne i}}^{n-1}W_{i,j}b_{i,j}\colon &&  \\
& \sum_{\substack{i=0\\i\ne j}}^{n-1} b_{i,j} = 1 && j=0, \ldots, n-1; \\
& \sum_{\substack{j=0\\j\ne i}}^{n-1} b_{i,j} = 1 && i=0, \ldots, n-1; \\
& u_i-u_j +nb_{i,j} \le n-1 && 1 \le i \ne j \le n-1;  \label{eq:app-ilp-ineq}
\end{align}
where $b_{i,j}$ is a binary variable that is equal to 1 iff we move from city $i$ to city $j$, and $u_i\in \{1,\dots,n-1\}$ are dummy variables. One can use the typical ways of transforming such ILP into QUBO~\cite{lucas2014ising,salehi2022unconstrained}, by first introducing the slack variables $\xi_{i,j}\in\{0,\dots,2n-3\}$ for each inequality in Eq.~\eqref{eq:app-ilp-ineq}, each using $\lceil \log (2n-3)\rceil \sim \log n$ bits, then moving the linear equality constraints into the objective function. Integer variables $u_i$ also require $\sim \log n$ qubits each. The final formula (we skip transforming variables into the combination of binaries for clarity) is
\begin{equation}
\sum_{i=0}^{n-1} \sum_{\substack{j\ne i\\j=0}}^{n-1}W_{i,j}b_{i,j} + A_1 \sum_{j=0}^{n-1} \bigg(\sum_{\substack{i=0\\i\ne j}}^{n-1} b_{i,j} - 1\bigg)^2 + A_2 \sum_{i=0}^{n-1} \bigg(\sum_{\substack{j=0\\j\ne i}}^{n-1} b_{i,j} - 1\bigg)^2 + A_3 \sum_{\substack{i,j=1\\i\neq j}}^{n-1} \left(u_i-u_j +nb_{i,j} - n+1 + \xi _{i,j} \right)^2.\label{eq:tsp-ilp-qubo}
\end{equation}
We use the same initial state and mixer as in Sec.~\ref{sec:app-maxcut-qaoa}.

The number of qubits is $\sim 2n^2 \log n$, which is the same as the effective space size. The number of 2-local Pauli terms is $\order{n^3}$ for both second and third constraints in Eq.~\eqref{eq:tsp-ilp-qubo}, and for the fourth it is $\binom{n-1}{2}\order{\binom{\log n}{2}} = \order{n^2 \log^2 n}$. Note that some Pauli terms were counted multiple times for the fourth constraint; however, the count is still much smaller than $\order{n^3}$ from the second and the third constraints.

Let us first consider the depth on the all-to-all architecture. The first part of Eq.~\eqref{eq:tsp-ilp-qubo} can be implemented in depth $1$. Second, using the order as in the swap network within each register $\{b_{i,j}:j=0,\dots,n-1, j\ne i\}$ for fixed $i$, the second part can be implemented in depth $\order{n}$. The third constraint is symmetric to the second, thus it can also be implemented in depth $\order{n}$.
For the fourth, the only variables shared between constraints are the dummy $u_{i}$ variables. Thus we choose the swap network order of interacting $u_i$ and $u_j$, and in addition, when interactions between $u_i$ and $u_j$ are implemented, we add remaining interactions with and within $b_{i,j}$ and $\xi_{i,j}$. Since each of the constraints is a full QUBO defined over $\order{\log n}$ qubits, in total, we will need $\order{n \log n}$ depth. Thus the total depth is $\order{n \log n}$, which is tight given the number of qubits and the number of gates to be implemented.

For LNN, we can use the swap network presented in~\cite{o2019generalized} to
achieve $\order{n^2\log n}$ depth. In fact, this is almost optimal by
Theorem~\ref{theorem:depth-lowerbound-lnn-constant-radius}, as the distance
between any pair of bits is $\order{1}$. The distance between bits $b_{i,j}$ and
$b_{k,l}$ is at most 2 because of the second and the third type of constraints.
The distance between $u$- or $\xi$- and any $b_{k,l}$ is again constant, as for
any $u$- or $\xi$-generated bit there is a neighboring $b_{i,j}$ bit. 

For the energy span, let us see that 
\begin{multline}
\sum_{i=0}^{n-1} \sum_{\substack{j\ne i\\j=0}}^{n-1}W_{i,j}b_{i,j} + A_1 \sum_{j=0}^{n-1} \bigg(\sum_{\substack{i=0\\i\ne j}}^{n-1} b_{i,j} - 1\bigg)^2 + A_2 \sum_{i=0}^{n-1} \bigg(\sum_{\substack{j=0\\j\ne i}}^{n-1} b_{i,j} - 1\bigg)^2 + A_3 \sum_{\substack{i,j=1\\i\neq j}}^{n-1} \left(u_i-u_j +nb_{i,j} - n+1 + \xi _{i,j} \right)^2  \\
\leq n^2 \|W\|_\infty + A_1n\cdot n^2 + A_2n\cdot n^2 + A_3 n^2 \cdot n^2 = \order{n^4}.
\end{multline}
Note that $\order{n^4}$ is achievable by the assignment $b_{i,j}\equiv 0$, $u_i\equiv 1$ and $\xi_{i,j} \equiv 2n-2$, hence our derivation is tight.

The initial state preparation and mixer consist of single-qubit gates only; thus, their costs are negligible.

\subsubsection{Prog-QAOA}

Implementing Prog-QAOA is equivalent to implementing the following function:
\begin{equation*}
A  [\{b_t\}\textrm{ is not a permutation}] + \sum_{t=0}^{n-1}\sum_{\substack{i,j=0\\i\neq j}}^{n-1} W_{i,j}\delta(b_t, i) \delta(b_{t+1}, j).
\end{equation*}
For our purposes, we choose the initial state similar to the one from Sec.~\ref{sec:sim-maxkcut} and Grover mixer~\cite{bartschi2020grover}.

We need $n\lceil \log n\rceil$ qubits to encode the cities. For permutation checking, we need $n+1$ additional qubits, $n$ for counting, and $1$ for $Z$-rotation. These qubits can be reused for objective function implementation, which requires $n\lceil \log n\rceil $ qubits to store the edge-based representation and $n$ qubits for the flags needed while creating the representation. Then we use the same flags for implementing the actual costs. So the total number is $2n\lceil \log n \rceil +n =\Theta(n \log n)$. The effective space size is $\log(n^n) = n\log n$.

Let us start with determining the depth and number of gates in all-to-all architecture. To implement the permutation checking, we need to implement $\order{n^2} $ unitaries for all pairs of quantum registers $(b_t,b_{t'})$ to count the occurrence of each city modulo 2. Each unitary requires a NOT gate controlled by $\sim\log n $ qubits, which can be implemented without using any ancilla in $ \sim\log n $ depth~\cite{saeedi2013linear,da2022linear}. Hence, this step can be done in depth $ \order{n \log n} $. Then we need a NOT gate controlled by $ n $ qubits, which can be implemented in $\order{n}$ depth using $ \order{n^2} $ gates. Uncomputation does not change the complexity.

For the objective part, we need to implement $ \order{n^2} $ unitaries to construct the edge representation through all-to-all strategy. Each unitary requires $ \order{\log n}$ depth and $\order{\log^2 n}$ gates to set the flag, $\order{\log n}$ depth and $\order{\log n}$ gates to implement the controlled copy, hence overall we need $ \order{n \log n} $ depth and $\order{n^2 \log^2n}$ gates. Implementation details can be found in Sec.~\ref{sec:xor} and~\ref{sec:controlled-copy}. The last step which implements the cost requires $ \order{n^2} $ unitaries each with $ \order{\log n} $ depth and $\order{\log^2n}$ gates. Note that these unitaries can be implemented in parallel giving $\order{n\log n}$ depth. We conclude that the overall depth is $\order{n\log n} $, and the number of gates is $\order{ n^2 \log^2 n}$.

In the LNN architecture, we need to use the procedures described in Theorems~\ref{theorem:reginterlace} and \ref{theorem:inter-all-to-all}. This puts an additional cost of swapping the registers which is $\order{\log n}$, which results in depth $ \order{n \log^2 n} $.

The number of parameterized gates is 1 for permutation checking, $\order{n^2}$ for objective part implementation, and $1$ for Grover Mixer using a single-ancilla decomposition from Lemma 7.11 from~\cite{barenco1995elementary}, hence $\order{n^2}$ in total. For the energy span, observe that
\begin{equation}
\sum_{t=0}^{n-1}\sum_{\substack{i,j=0\\i\neq j}}^{n-1} W_{i,j}\delta(b_t, i) \delta(b_{t+1}, j) + A  [\{b_t\}\textrm{ is not a permutation}]  \leq  n\max_{i\neq j}W_{i,j} + A 
= \order{n}.
\end{equation}

Similarly, as it was in Sec.~\ref{sec:sim-maxkcut}, the cost of the initial state and mixer is negligible.

\section{Superposition of permutations in binary encoding}\label{sec:permutations}

Let us recall the original algorithm for generating superposition of permutation in one-hot encoding~\cite{bartschi2019deterministic}. The circuit for generating the superposition of permutations uses $n$ registers, each consisting of $n$ qubits. At each step, one of the registers is fixed so that it complies with the superposition of permutations requirements (the details of the circuit can be found in \cite{bartschi2019deterministic}). The first step fixes the first register, using the second register called the \emph{mask register}. The first step ends with swapping the second and third registers, thus creating space for the \emph{workspace register}.

The next $n-3$ steps fix the consecutive registers as follows. Workspace and mask registers are moved down by $n$ qubits to open up room for the next register to be fixed (so that the new register is between the already fixed registers and the workspace register). Then an appropriate unitary operation is applied so that this new sandwiched register is fixed. 

The final step fixes the workspace and mask registers to obtain the desired final state.

In the case of the binary encoding, the following is changed: at the end of each step, once a register is fixed, a unary-to-binary encoding circuit is applied as in Fig.~\ref{fig_utob}. As we will show later, encoding change can be accomplished with depth $\order{n}$ on LNN architecture. Since the output qubits storing the value in binary encoding are scattered in the register, they are moved to the top of the register leaving $n-\lceil\log n\rceil$ free $\ket{0}$ qubits below. Then, the workspace and mask registers need to be moved down only by $\lceil \log n \rceil $ qubits to create an empty space of $n$ qubits for the next step. Finally, in the end, after the workspace and map registers are fixed, we apply unary-to-binary encoding on both of them and move all the relevant qubits to the top so that the first $n\lceil \log n\rceil$ qubits store the superposition of permutations. Note that since these qubits need to be moved through $\order{n}$ qubits, the overall depth is still $\order{ n}$. Overall, we need $3n + (n-3) \lceil \log n \rceil = \order{n\log n} $ qubits, of which the last $3(n-\log n)$ qubits are set to $\ket{0}$ at the end, and those are the ancilla qubits required for the procedure.

\begin{figure}\centering 
	\includegraphics{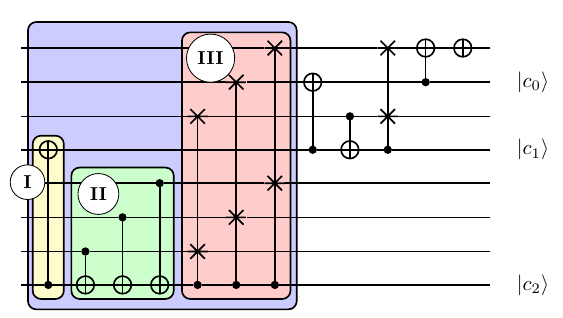}
	\caption{\label{fig_utob}Quantum circuit transforming the number in unary encoding into equivalent binary encoding~\cite{sawaya2020resource}. $\ket{c_0}$, $\ket{c_1}$, $\ket{c_2}$ are the output bits, and the qubits without labels in the output are set to $\ket{0}$. Parts I, II, and III are applied repeatedly for consecutive bits in binary encoding.}
\end{figure}

Let us now show that the procedure presented in Fig.~\ref{fig_utob} can be implemented in $\order{n}$ depth even on the LNN network. Note that the procedure consists of several similar steps applied on smaller and smaller circuits, a single step distinguished with a blue block. Each such step consists of three parts: a single CNOT (yellow) which can be implemented by shifting the target qubit down to the one-by-last qubit, applying CNOT, and then undoing the shift. Then there is a sequence of CNOTs represented by the green block, in which the target qubit should be moved upwards after each CNOT and should be moved back to the bottom at the end. Note that in both cases we have $\leq CN$ gates, where $C$ is a constant (in the next iterations, it will be $CN/2$, $CN/4$, etc.).

The most complicated procedure is the one distinguished by the red block, where we have a sequence of controlled-SWAPs. In there, we start by interlacing all the qubits. Note that after the interlacing, all qubits on which the swaps need to be applied are neighboring. So we only need to move control qubits through the whole register in a similar fashion, as it was done in the circuit described in Sec.~\ref{sec:controlled-copy}. Then the control qubit should be moved to its original position, and the interlacing should be undone. Note that even after these updates, the circuit can be implemented in $\order{ n}$ depth.

Since all the adjustments to the encoding change circuit do not increase the depth at each step, the total depth (and the number of gates) remains the same as in~\cite{bartschi2019deterministic}, namely $\order{n^2}$ depth and $\order{ n^3}$ gates on LNN.

\section{Additional examples} \label{sec:additional-examples}

\subsection{Integer Linear Program}

\paragraph{Problem:} An integer linear program (ILP) consists of an objective function to be minimized and constraints that need to be satisfied. The goal is to find the integer solution $\{x_i\}_{i=1}^m$ that minimizes the objective function. Formally, an ILP is expressed as follows:
\begin{align}
&\text{minimize} &\sum_{i=1}^{m} \sum_ia_i x_i  &&  \\
&\text{subject to}& \sum_{i=0}^m b_i^{(j)} x_i \leq c^{(j)}&& j\in J_1; \\
&& \sum_{i=0}^m b_i^{(j)} x_i = c^{(j)} && j\in J_2;\\
&& 0 \leq x_i < \bar x_i
\end{align}
Here, $a_i^{(j)}, b_i^{(j)}$ and $ c^{(j)}$ specify the problem instance. Note that we do not assume these parameters must be nonzero but we assume that they are small integers, requiring at most $C\log m$ bits for their representation, for some $C\geq 0$. $J_1$ and $J_2$ are disjoint sets of labels s.t. $|J_k|=m_k$. Finally, let $n_j$ for $j\in J_1\cup J_2$ be the number of nonzero parameters $b_i^{(j)}$. 

\paragraph{Solution:}
We take $m$ registers to store variables $x_i$ in binary encoding, each having $ \log \bar x_i$ qubits. Furthermore, we need two additional registers: one for computing $p=b_i^{(j)}x_i$, and another register for computing $\xi=c^{(j)} - \sum_{i=0}^m b_i^{(j)} x_i$. The latter value is computed as in Alg.~\ref{alg:inequality}. Note that for an equality constraint, it is enough to just return the value $\xi$ which is equal to 0 iff the equality is satisfied and nonzero otherwise. The number of qubits required for registers $\xi$ and $p$ should be sufficiently large to store an arbitrarily large product, but in typical scenarios where $c^{(j)}$, $b_i^{(j)}$, $a_i$ and $\bar x_i$ are $\order{\poly(m, |J_1|, |J_2|}$, the number of qubits grows only logarithmically with the size of the data per variable. 

The objective value can be computed similarly. 

\paragraph{Pros and cons:} Traditional methods of transforming the problem into a QUBO~\cite{salehi2022unconstrained} requires introducing a slack variable for each inequality to turn it into an equality, which increases the number of qubits required proportionally to $|J_1|$. Our method requires (under reasonable assumption) only $\order{\polylog(m, |J_1|, |J_2|}$ additional qubits. 

In addition, slack variables add redundant degrees of freedom, as for each inequality there is at most one value of the slack variable that makes the inequality satisfiable. This makes the problem more complicated in general. Since our method does not require explicit optimization of slack variables, we claim it is easier to optimize.

With a typical QUBO formulation, the penalty can be as large as the maximum attained value of $(\sum_{i=0}^m b_i^{(j)} x_i - c^{(j)} + \xi)^2$, which can grow quadratically with the number of variables times the maximum possible product $b_i^{(j)} x_i$. In our case, the largest energy is limited by $|\sum_{i=0}^m b_i^{(j)} x_i - c^{(j)}|$, which gives a quadratic improvement in the number of shots required for energy estimation. Furthermore, one may choose even a more extreme penalization with $[\xi < 0]$ for inequalities, which results at most $\order{1}$ up to a rescaling factor. However, in the latter case, a particularly large penalty value might be required in order to make the penalty constant meaningful w.r.t. to the objective function.

Note that while our method is very efficient in the number of qubits, it is particularly costly in depth, as up to poly-logarithmic factors it grows proportionally with the number of non-zero $b_{i}^{(j)}$. This issue can be overcome by adding more $p$-like and $\xi$-like registers and parallelizing $\xi$ computation. This in general depends on the sparsity and form of the inequalities, and optimal design of such parallelization goes beyond the scope of this paper. We would like to note that some possible improvements for the problem can be found in~\cite{de2019knapsack,gilliam2021grover} as well, who contributed with a similar approach to the introduced one.

\begin{algorithm}[h]
	\caption{Program for computing validating linear inequality .}
	\label{alg:inequality}
	\begin{algorithmic}[]
		\REQUIRE $(b_1^{(j)},\dots,b_m^{(j)}, c^{(j)})$ -- coefficients for linear inequality, $(x_1,\dots,x_m)$ -- variables from the inequality
        \STATE $\xi \gets c^{(j)}$
		\FOR{$i=1$ to $m$} 
        \STATE $p \gets b_i^{(j)}x_i $ \hfill // This has to be done only for $b_i^{(j)}\neq 0$
        \
		\STATE $ \xi \gets \xi - p$
        \STATE \textbf{uncompute} $p \gets b_i^{(j)}x_i $
		\ENDFOR
        \IF{ $\xi < 0$}
            \STATE $\xi' \gets |\xi|$ \hfill // Can be replaced with e.g. $\xi^2$
        \ELSE
            \STATE $\xi' \gets 0$
        \ENDIF
		\RETURN $\xi'$
	\end{algorithmic}
\end{algorithm}

\subsection{Travelling Salesman Problem with Time Windows}
\paragraph{Problem:} The Travelling Salesman Problem with Time Windows (TSPTW) is a generalization of the original TSP, in which each city must be visited within a given time window. Besides the usual travel cost $W_{uv}$ between the cities which can be interpreted as the time it takes to move from city $u$ to city $v$, each city $v$ is associated with an earliest start time $e_v$ and a latest due time $l_v$. If the salesman arrives at city $v$ before $e_v$, he has to wait. TSPTW aims to find a permutation of the cities that minimizes the cost and satisfies the time window constraints. 

\paragraph{Solution}
The objective function and permutation constraint are implemented as for the usual TSP. The only addition we need is a program, that checks whether the time windows constraints are satisfied for each city, for which we present a possible solution in Alg.~\ref{alg:tsptw}. First, we change the representation from time-to-city to city-to-city, which can be done identically to the transformation in the objective function, and store the next visited city within variables $\{\text{edge}_i\}$. This time, we also keep track of the visiting order of the cities within variables $\{\text{ord}_i\}$ (city-to-time formulation). Using the $\text{edge}_i$ variables, we construct a set of variables $\textrm{cost}_i$ that holds the cost between the cities $i$ and $\text{edge}_i$. Furthermore, we create an additional set of variables $\textrm{ord\_cost}_i$ to store the ordered list of the costs with respect to time i.e. to store the cost of the edge traversed at time $i$. For this procedure, we need $4n$ additional new registers, where the registers in the first two sets need $ \log{n}$ qubits, while the registers from the other two sets need to be able to represent numbers up to arbitrary precision since they contain the cost of moving between the cities. Having the time-ordered list of costs in $\{\textrm{ord\_cost}_i\}$, we can calculate the arrival times for each city. If the arrival is before the earliest start time of the city, then the waiting time is incorporated by setting the time as the maximum of the earliest start time and the arrival time. This is stored in $\{\text{time}_t\}$ and gives the actual time (also called the service time) the city at position $t$ is visited. This time is used for checking if the latest due time is violated for a particular city. In the case the latest due time constraint is violated, we can either just yield $1$ each time this happens, or we can yield the difference between the latest start time and the service time of the corresponding city, as a measure of the severity of the violation. For the latter, we need additional $n+2$ registers, that can store arbitrary precision.

\paragraph{Pros and cons:}

The time-to-city representation requires $n\log{n}$ qubits. For the representation change (variables edge$_i$ and ord$_i$), we need $2 n \log{n}$ additional qubits. Then, to construct the time-ordered list of costs, we need further $2 n$ registers, all of which can express numbers to required precision using floating-point arithmetics. Normally that means that $\order{\log n}$ per register is required for variables cost$_i$ and ord\_cost$_t$. Note that $l_u =  \order{n}$ under this assumption so we only need additional $\order{n \log n}$ qubits for each variable tim$_t$, and for new\_time. Finally, we need at least one qubit for the flag. With these assumptions, the number of qubits is $\order{n \log n}$ in total. Similarly to the basic TSP, the corresponding quantum circuit requires only $\tildorder{n^2}$ gates to implement. Note, that this is better than all the known encodings of even the basic TSP shown in Tab.~\ref{table:tsp}. To the best of our knowledge, the most efficient formulations of this problem so far are presented in \cite{salehi2022unconstrained}. In this work, the authors present several formulations of this problem. The edge-based formulation uses $\order{n^3}$ qubits and $\tildorder{n^5}$ gates. The node-based formulation uses $\order{n^2}$ qubits and $\tildorder{n^5}$ gates. Finally, the ILP based formulation require $\order{n^2}$ qubits and $\order{n^3}$ gates.

\begin{algorithm}[t]
	\caption{Program for verifying the time-windows constraints in TSPTW}
	\label{alg:tsptw}
	\begin{algorithmic}[]
		\REQUIRE $(c_0,\dots,c_{n-1})$ -- the list of integers denoting visited cities, ${W}$ -- set of arcs, $e_v$, $l_v$ -- set of earliest start times and latest due times
            \STATE $\textrm{flag} \gets 0 $
		\FOR{$i=0$ to $n-1$} 
		\FOR{$t=0$ to $n-1$} 
		\STATE $\textrm{flag} \gets c_{t}= i$
		\IF{$\textrm{flag} = 1$} 
		\STATE $\textrm{edge}_i \gets c_{t+1}$ 
            \STATE $\textrm{ord}_i \gets t$
		\ENDIF 
        \STATE \textbf{uncompute} $\textrm{flag} \gets c_{t}= i$
		\ENDFOR
		\ENDFOR
  
            \FOR{$i=0$ to $n-1$} 
		\FOR{$j=0$ to $n-1$} 
		\IF{$\textrm{edge}_i = j$} 
		\STATE $\textrm{cost}_i \gets {W}_{ij}$ 
		\ENDIF 
		\ENDFOR
		\ENDFOR 

            \FOR{$t=0$ to $n-1$} 
		\FOR{$i=0$ to $n-1$} 
		\IF{$\textrm{ord}_i = t$} 
		\STATE {$\textrm{ord\_cost}_t \gets \textrm{cost}_i$}
		\ENDIF 
		\ENDFOR
		\ENDFOR 
  
            \STATE $\textrm{new\_time} \gets 0 $
            \STATE $\textrm{$\textrm{time}_0$} \gets 0 $
		\FOR{$t=0$ to $n-1$} 
            \STATE $\textrm{$\textrm{time}_{t+1}$} \gets 0 $
		\FOR{$i=0$ to $n-1$} 
		\IF{$ \textrm{\rm ord}_{i} = (t+1)\% n$}
            \STATE $ \textrm{new\_time}  \gets  \textrm{$\textrm{time}_{t}$} + \textrm{ord\_cost}_t$
            \STATE $ \textrm{$\textrm{time}_{t+1}$}  \gets  \max{(\textrm{new\_time}, e_i)}$
            \STATE \textbf{uncompute} {$ \textrm{new\_time}  \gets  \textrm{$\textrm{time}_{t}$} + \textrm{ord\_cost}_t$}
            \STATE $ \textrm{flag}  \gets  \textrm{$\textrm{time}_{t+1}$} \geq l_i$
            \STATE \textbf{yield} flag \hfill //And apply $Z$ rotation
            \STATE \textbf{uncompute} {flag $\gets$ $\textrm{time}_{t+1}$ $\geq l_i$}
            \ENDIF
            \ENDFOR
		\ENDFOR 
		
	\end{algorithmic}
\end{algorithm}

\subsection{Special cases of general-metric unit disk graph isomorphism}
\paragraph{Problem:} 
Let $V$ be a finite set of points from a metric space $(A,d)$. Suppose graph $G=(V,E)$ is constructed following
\begin{equation}
    \{v,w\} \in E \iff d(v,w) \leq 1.
\end{equation}
We will call such graphs general-metric unit disk graphs. Given two such graphs $G_1=(V_1 ,E_1)$ and $G_2=(V_2,E_2)$ with $n\coloneqq |V_1|=|V_2|$, the goal is to find an isomorphism  $\varphi:V_1 \to V_2$, i.e. a bijective function between $V_1$ and $V_2$ s.t.
\begin{equation}
    \{v_1, v_2\} \in E_1 \iff \{\varphi(v_1), \varphi(v_2)\} \in E_2. \label{eq:graph-isomorphism-constraint}
\end{equation}
General graph isomorphism is known to be an NP problem. However, it is not known if it is P, NP-Hard, or NP-intermediate. Some variants of the problem, like fixed-dimensional Euclidean metric can be shown to lie in P~\cite{arvind2014complexity}.

Since the problem we consider is a special case of 
general graph isomorphism, one could use the implicit representation through only the set of points $V_1,V_2\subset A$ instead of an explicit representation of graphs $(V_1,E_1)$ and $(V_2,E_2)$. In this section, we assume that the implicit representation is more efficient, i.e. elements of $A$ can be sparsely written and the metric $d$ is efficient to compute.  This is satisfied for many common metrics defined over linear spaces, like Euclidean or Manhattan metrics. To be more explicit, our complexity analysis will hold iff the number of qubits required to represent an element of $A$ requires only $\polylog(n)$ qubits and $d(x,y)$ can be computed using $\polylog(n)$ gates. 

\paragraph{Solution:} We require $n$ registers $\varphi_v$ for $v\in V_1$ each having $\sim \log n$ qubits to store permutation $\varphi$. We will use binary encoding to store the permutation. Verification of whether $\varphi$ is a bijection can be done with Algorithm~\ref{alg:permutation}. Next, we need to check whether it satisfies Eq.~\eqref{eq:graph-isomorphism-constraint}. The program for this verification is presented in Algorithm~\ref{alg:graph-iso} and goes as follows. First, in separate registers $\textrm{coord}_v$, we store a representation of $\varphi(v) \in A$, for instance for 2-dimensional space, it will be $\textrm{coord}_v = (x,y)$. Then for each of the 2-combinations of $V_1$, we compute the distance between the points of the corresponding nodes in $V_2$ and set the flag to 1 if the nodes are connected according to the metric and 0 otherwise. If $\{v_1,v_2\}$ are in $E_1$, then the value of the flag is negated so that it will give 0 if both graphs have or do not have the corresponding edge, and 1 otherwise. Finally, the value of the flag is uncomputed. Note that the check $\{u,v\} \in E_1$ is performed on the classical computer, and does not require from $G_1$ to be represented through the implicit form of the set $V\subset A$.

We need additional registers $\{\textrm{coord}_v\}_{v \in V_1}$ to store the coordinate of the corresponding vertex $b(v) \in V_2$. For instance, for $k=2$ dimensional space, $\textrm{coord}_v = (x_{b(v)},y_{b(v)})$. This can be done in parallel for each pair of registers $(b_v, \textrm{coord}_v)$. Then for each of the 2-combinations $(u,v)$ of $V_1$, we compute the distance between the corresponding nodes in $V_2$ and set the flag to 1 if the nodes are connected according to the metric and 0 otherwise. If $\{u,v\} \in E_1$, then the value of the flag is negated so that it will give 0 if both graphs have or do not have the corresponding edge, and 1 otherwise. Finally, the value of the flag is uncomputed.

\paragraph{Pros and cons:}
The method requires only $\tildorder{n}$ qubits, as $b_i$ registers take all-together $\sim n \log n$ qubits, $\textrm{coord}_v$ takes $\tildorder{k\log n}$ qubits, dist takes $\order{\log n}$ qubits (similarly addition and multiplication may require $\order{\polylog(n)}$ auxiliary qubits and finally flag takes one qubit. The corresponding quantum circuit requires $\tildorder{n^2}$ gates. Note that the first loop runs in parallel for all pairs of registers $(b_v, \textrm{coord}_v)$ thus the corresponding circuit has depth $\tildorder{n}$. For the next loop, using round-robing tournament scheduling one can also achieve $\tildorder{n}$ depth, which gives total depth $\tildorder{n}$. All encodings known to the authors requires $\order{n^2}$ qubits and have $\order{n^3}$ terms~\cite{lucas2014ising, chatterjee2021variational, calude2017qubo, zick2015experimental} in general case scenario. Note that in~\cite{zick2015experimental} an attempt to a significant reduction in the number of variables is made, however; it is unclear how it will work for the general Euclidean case.

\begin{algorithm}[t]
	\caption{General-metric unit-disk graph isomorphism cost.}\label{alg:graph-iso}
	\begin{algorithmic}[]
		\REQUIRE $\varphi_i$ -- register for storing the value $\varphi(i)$; $\{p_w: w \in V_2\}$ classically known coordinates of vertices from $V_2$ \\[1ex]
		\FOR {$v \in V_1$}
		\FOR {$w\in V_2$} 
		\IF {$\varphi_v = w$}
		\STATE $\textrm{coord}_v \gets p_w$ 
		\ENDIF
		\ENDFOR
		\ENDFOR 
		\FOR{$(u, v) \in \textrm{ 2-combinations of }V_1$}
		\STATE $\textrm{dist} \gets d(\textrm{coord}_{u}, \textrm{coord}_{v}) $ \hfill //distance between $b_{u}$ and $b_{v}$
		\STATE $\textrm{flag} \gets \textrm{dist} \leq 1$
		\IF {$\{u,v\} \in E_1$}
            \STATE $\textrm{flag} \gets  \neg \textrm{flag}$
		\ENDIF
        \STATE \textbf{yield} \textrm{flag}
        \IF {$\{u,v\} \in E_1$}
            \STATE \textbf{uncompute} $\textrm{flag} \gets  \neg \textrm{flag}$
        \ENDIF
        \STATE \textbf{uncompute} $\textrm{flag} \gets \textrm{dist} \leq 1$
		\STATE \textbf{uncompute} $\textrm{dist} \gets d(\textrm{coord}_{u}, \textrm{coord}_{v}) $ 
		\ENDFOR 				
	\end{algorithmic}
\end{algorithm}

\section{Numerical experiments}

\subsection{Quantum circuit cost estimation for Max-$K$-Cut}\label{sec:maxkcut-numerics}

For the results presented in Fig.~\ref{fig:maxkcut-numerics}, we implemented each considered QAOA using the \texttt{qiskit} package. We made two versions of each algorithm, one for all-to-all architecture and another for LNN. 
We used swap networks like register swap, interlacing, and intra-all-to-all strategies as described in Appendix Sec.~\ref{sec:swap-networks}.
To implement the multi-controlled NOT operator, we considered two approaches. For all-to-all architecture, we used the v-chain implementation of qiskit, which uses $l-1$ ancillas, where $l$ is the number of control qubits. For LNN, we used the no-ancilla implementation presented in \cite{da2022linear} with the LNN-friendly modification taken from \cite{saeedi2013linear}. Once created, the circuit was transpiled into a circuit consisting of only general 1-qubit gates and CNOT gates. In the case of LNN architecture, the transpilation included information about the restricted qubit connectivity, however, our implementation of the quantum circuit was already adjusted to LNN prior to transpilation. For the transpilation, we used optimization level $3$. The estimations presented include the cost of preparing the initial state and a single level of QAOA. The number of gates and depth include both 1- and 2-qubit gates. 

We considered the Max-$K$-Cut problem for a complete weighted graph with $n=128$ nodes and $K=3,4,\dots,63$ colors. Note that quality metrics in our experiment depend only on the chosen QAOA variant, the number of colors, and the number of nodes and do not depend on the weights.

\subsection{TSP optimization}\label{sec:tsp-numeric}
All algorithms were implemented in the Julia programming language. For each QAOA variant and each number of cities, we considered 50 TSP instances defined on a directed complete graph with the elements of cost matrix sampled i.i.d. uniformly from $\{1,\dots,10\}$. Penalty parameters for the TSP over $n$ cities were chosen to be $2n\max_{i\neq j}W_{i,j}$ for Prog-QAOA and $2\max_{i\neq j} W_{i,j}$ for other algorithms. 

We used the L-BFGS optimizer implemented in Optim.jl. For the first 5 levels, for each TSP instance, 5 optimization runs were considered, starting each time with a random sequence of angles drawn, each initial vector sampled from the uniform distribution $[0,\pi]$ for the mixer angles and $[0,2\pi]$ for the objective Hamiltonian angles. The optimization was run on a hyperrectangle, where the length of each dimension is determined by the domain of the corresponding angle with periodic boundary conditions. This is effective because the energy function in terms of angles is periodic.

For $k>5$ levels, the optimized vector from the previous step was used as the initial vector, enlarged with random angles for the newly added $k$-th level drawn from the same distribution as it was for the initial vector for $k\leq 5$. During the optimization, the whole vector is optimized, not only the two newly added angles. The stopping conditions were chosen to be a change of $10^{-5}$ of the optimized vector \texttt{x\_tol}, gradient \texttt{g\_tol} or the change of the objective value \texttt{f\_tol}. We run the experiment 5 times per each QAOA variant and each TSP instance, resulting in 5 trajectories of energies. For each level, only the best result was chosen for visualization, independently for the probabilities and the rescaled energy. The area presented in the figure reflects the samples within the standard deviation from the mean. For the second and the third rows, the areas were omitted to improve readability.

We would like to mention that XY- and Prog-QAOA with cyclic mixer on each register and Prog-QAOA with $X$ mixer were also analyzed in our experiments. However, the results were clearly worse and thus were omitted. We were not able to consider ILP-originating QAOA in general and $X$-QAOA for 5 and 6 cities due to the unattainable qubit requirements.

\end{document}